\documentclass[11pt,letterpaper]{article}

\usepackage[margin=1in]{geometry}
\usepackage{amsfonts, amsmath, amssymb, amsthm}
\usepackage{thmtools, mathtools}
\usepackage{nicefrac}
\usepackage{enumitem}
\usepackage{soul}

\declaretheorem{theorem}
\declaretheorem[sibling=theorem]{fact}

\declaretheorem[sibling=theorem]{lemma}
\declaretheorem[sibling=theorem]{corollary}

\declaretheorem[style=definition]{remark}

\usepackage[dvipsnames,usenames]{xcolor}
\usepackage[colorlinks=true,pdfpagemode=UseNone,urlcolor=RoyalBlue,linkcolor=RoyalBlue,citecolor=OliveGreen,pdfstartview=FitH]{hyperref}

\usepackage{cleveref}
\usepackage[round]{natbib}

\usepackage{booktabs}
\usepackage{multirow}

\usepackage{tcolorbox}
\usepackage{todonotes}

\newtcolorbox{algorithm}[1]
{
	adjusted title = {#1},
	fonttitle = \bfseries,
  	beforeafter skip = 12pt,
}

\usepackage{subcaption}
\usepackage{tikz,pgfplots}

\usetikzlibrary{decorations.pathreplacing}
\usepgfplotslibrary{fillbetween}
\pgfplotsset{compat=1.15}

\newcommand{\Distribution}{\mathcal{D}}
\newcommand{\maximize}{\mathrm{maximize}}
\newcommand{\minimize}{\mathrm{minimize}}
\newcommand{\subjectto}{\mathrm{subject\ to}}

\newcommand{\dif}[1]{\,\mathrm{d}#1}

\newcommand{\CP}{\mathrm{CP}}

\newcommand{\OPT}{\mathrm{OPT}}
\newcommand{\ALG}{\mathrm{ALG}}

\newcommand{\E}{\operatorname{\mathbf{E}}}
\renewcommand{\Pr}{\operatorname{\mathbf{Pr}}}

\newcommand{\argmin}{\operatorname{\arg\min}}

\newcommand{\defeq}{\stackrel{\mathrm{\footnotesize def}}{=}}

\title{Optimal $4$-Approximation for the Correlated Pandora's Problem}
\author{
    Nikhil Bansal
    \thanks{University of Michigan. Email: bansaln@umich.edu. Supported in part by the NSF awards CCF-2327011 and CCF-2504995.}
    \and
    Zhiyi Huang
    \thanks{The University of Hong Kong. Email: zhiyi@cs.hku.hk, zhuzixua@connect.hku.hk.}
    \and
    Zixuan Zhu
    \footnotemark[2]
}

\begin{document}

\begin{titlepage}
    \thispagestyle{empty}
    \maketitle
    \begin{abstract}
        \thispagestyle{empty}
        The Correlated Pandora's Problem posed by Chawla et al.~(2020) generalizes the classical Pandora's Problem by allowing the numbers inside the Pandora's boxes to be correlated.
        It also generalizes the Min Sum Set Cover problem, and is related to the Uniform Decision Tree problem.
        This paper gives an optimal $4$-approximation for the Correlated Pandora's Problem, matching the lower bound of $4$ from Min Sum Set Cover.
    \end{abstract}

    \clearpage

    \thispagestyle{empty}
    \setcounter{tocdepth}{2}
    \tableofcontents
\end{titlepage}
\allowdisplaybreaks
\section{Introduction}
\label{sec:intro}

Information plays a crucial role in high-stakes decision-making, but is often costly to acquire.
Consider a technology company aiming to train its next-generation large language model (LLM) to achieve exceptional reasoning abilities.
The company faces key decisions: 
Should they use traditional dense transformers, the recently successful Mixture-of-Experts (MoE) approach, or a Hybrid solution? 
Should they adopt Supervised Fine-Tuning (SFT), Reinforcement Learning (RL), or both in post-training? 
These choices must be made before committing to the final training phase, which requires months of time and billions of dollars. 
While the company can run small-scale tests to estimate performance metrics, two questions emerge: 
Which architectural combinations should they test, and how do they determine when they have tested enough?

The Pandora's Problem by \citet{Weitzman/1979/original} is a classical model of decision-making with costly information acquisition.
It considers a set of Pandora's boxes (e.g., combinations of architectural decisions).
Inside each box is an unknown volume (e.g., the time of the training phase with this combination) drawn {\em independently} from some distribution.
Opening a box reveals the volume inside, but comes at a cost (e.g., the time required for a small-scale test).
At any time, the algorithm decides whether to open another box or to stop and take the minimum volume from an opened box.
We want to minimize the sum of the box-opening costs and the taken volume (e.g., to release the LLM as early as possible).

\medskip
\noindent {\bf Correlated Pandora's Problem.~}
Comparing this classical model with our motivating scenario, the assumption of independently drawn volumes stands out as a limitation. 
For example, if we find that the combination of MoE and SFT works poorly, then other combinations involving MoE or SFT are likely to underperform as well. 
\citet{ChawlaGTTZ:FOCS:2020} proposed the Correlated Pandora's Problem, where the vector of volumes of the boxes, referred to as the \emph{scenario}, is drawn from a possibly correlated distribution. We define this formally in Section \ref{sec:prel}.

Besides its practical relevance, the Correlated Pandora's Problem is also interesting for its connection with other optimization problems, including the Min Sum Set Cover (MSSC) problem \citep{FeigeLT:Algorithmica:2004} and the Uniform Decision Tree (UDT) problem \citep{ChakaravarthyPRAM:PODS:2007}.

\medskip
\noindent 
{\bf Connection to MSSC.~} We elaborate on the connection to MSSC, as this will play an important role.
In MSSC, we are given a universe $U$ of elements and a collection of subsets $S_1, S_2, \dots, S_n$ of $U$.
The union of the sets $S_i$ is $U$. 
We want to order these sets one at a time to cover all the elements, such that the average cover time of elements is minimized.\footnote{An element $e$ is covered at time $t$, if $e$ lies in the set at position $t$ in the ordering, and not in any earlier set.}

MSSC can be seen as a special case of the Correlated Pandora's Problem. 
Given any MSSC instance, let there be a box $i$ for each subset $S_i$ and a scenario for each element $e \in U$.
The cost of opening each box is $1$.
Further, box $i$'s volume in scenario $e$ is $0$ if $e \in S_i$ and $\infty$ otherwise.
Finally, we consider the uniform distribution over the scenarios.
Since the volumes are either zero or infinite, we take the first zero-volume box in each scenario $e$, and the number of boxes opened by then exactly equals the cover time of element $e$.

\subsection{Previous Approaches Using Combinatorial Arguments}
\cite{FeigeLT:Algorithmica:2004} studied MSSC and considered the natural Greedy algorithm, which at each step, picks a subset that covers the most uncovered elements.
They showed that Greedy gives a $4$-approximation, and this approximation ratio is the best possible for polynomial time algorithms, assuming $\mathrm{P} \neq \mathrm{NP}$.

We briefly describe their combinatorial amortized analysis, which has served as the foundation for the positive results on the Correlated Pandora's Problem before this paper. 
First, they consider the elements' cover time in the optimal solution in descending order, as shown in \Cref{fig:optimal-opening}.
Next, they define an appropriate price for each element as the amortized objective by Greedy---%
the elements' averaged covered time equals their averaged price.
Consider the prices in the reverse order of when the elements are covered by Greedy, as shown in \Cref{fig:alg}.%
\footnote{The orders of elements in \Cref{fig:optimal-opening,fig:alg} could be different.}
Using an elegant argument, they showed that the price of the $i$-th element in \Cref{fig:alg} is at most twice the cover time of the $\frac{i}{2}$-th element in \Cref{fig:optimal-opening}.
This yields a $4$-approximation.

\citet{ChawlaGTTZ:FOCS:2020} solved the Correlated Pandora's Problem in two steps. 
First, they order the boxes by reducing to MSSC, rounding the volumes to $0$ or $\infty$ with carefully chosen thresholds. 
Although they used a linear program (LP) similar to the ones in this paper to decide the thresholds, their order of boxes and guarantees are given by the Greedy algorithm of \citet{FeigeLT:Algorithmica:2004} and the above combinatorial argument.
Finally, they combine the order of boxes with a randomized Ski Rental stopping algorithm to get a $9.22$-approximation.

Recently, \citet{GergatsouliT:NeurIPS:2023} 
extended Weitzman's Rule to the Correlated Pandora's Problem and obtained the state-of-the-art approximation ratio of $4.43$.
Their analysis is a more sophisticated variant of \citet{FeigeLT:Algorithmica:2004}'s analysis.
Besides the scenarios' cover time in the optimal solution (\Cref{fig:optimal-opening}) and the algorithm's amortized objective (\Cref{fig:alg}), they also considered the volumes that the optimal solution takes in different scenarios in descending order (\Cref{fig:optimal-volume}).
They showed that the amortized objective of the $i$-th scenario in \Cref{fig:alg} is at most a weighted sum of the $\alpha i$-th scenario's cover time in \Cref{fig:optimal-opening} and the $\beta i$-th scenario's volume in \Cref{fig:optimal-volume}, for some appropriate weights and parameters $\alpha$ and $\beta$.

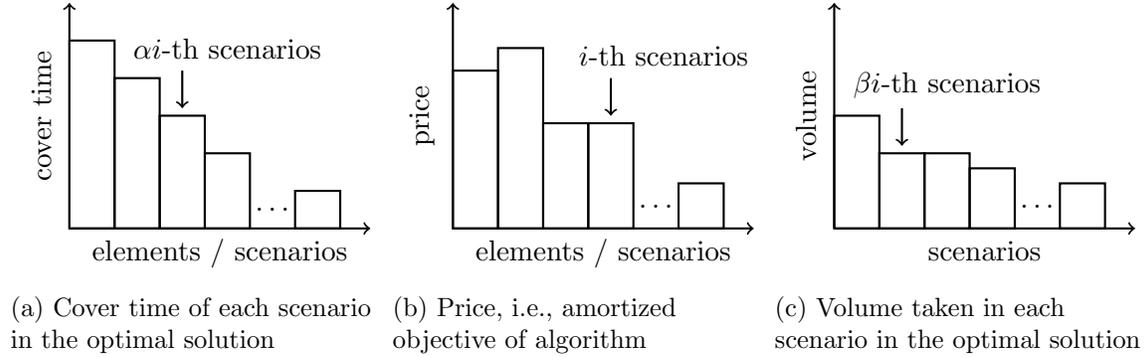
\begin{figure}
\centering
\begin{subfigure}{.3\textwidth}
\centering
\begin{tikzpicture}
	\draw[thick] (0,0) rectangle (0.6,2.5);
	\draw[thick] (0.6,0) rectangle (1.2,2);
	\draw[thick] (1.2,0) rectangle (1.8,1.5);
	\draw[thick] (1.8,0) rectangle (2.4,1);
	\node at (2.73,0.25) {\dots};
	\draw[thick] (3.0,0) rectangle (3.6,0.5);
	\draw[thick,->] (0,0) -- (0,3.0) node[midway,above=0.1cm,sloped] {cover time};
	\draw[thick,->] (0,0) -- (4,0) node[midway,below] {elements / scenarios};	
	\draw[thick,<-] (1.5,1.6) -- (1.5,2.1);
	\node at (2.1,2.4) {$\alpha i$-th scenarios};
\end{tikzpicture}
\caption{Cover time of each scenario\\ in the optimal solution}
\label{fig:optimal-opening}
\end{subfigure}	
\begin{subfigure}{.3\textwidth}
\centering
\begin{tikzpicture}
	\draw[thick] (0,0) rectangle (0.6,2.1);
	\draw[thick] (0.6,0) rectangle (1.2,2.4);
	\draw[thick] (1.2,0) rectangle (1.8,1.4);
	\draw[thick] (1.8,0) rectangle (2.4,1.4);
	\node at (2.73,0.3) {\dots};	
	\draw[thick] (3.0,0) rectangle (3.6,0.6);
	\draw[thick,->] (0,0) -- (0,3.0) node[midway,above=0.1cm,sloped] {price};
	\draw[thick,->] (0,0) -- (4,0) node[midway,below] {elements / scenarios};	
	\draw[thick,<-] (2.1,1.5) -- (2.1,2.0);
	\node at (2.8,2.3) {$i$-th scenarios};
\end{tikzpicture}
\caption{Price, i.e., amortized\\ objective of algorithm}
\label{fig:alg}
\end{subfigure}	
\begin{subfigure}{.3\textwidth}
\centering
\begin{tikzpicture}
	\draw[thick] (0,0) rectangle (0.6,1.5);
	\draw[thick] (0.6,0) rectangle (1.2,1);
	\draw[thick] (1.2,0) rectangle (1.8,1);
	\draw[thick] (1.8,0) rectangle (2.4,0.8);
	\node at (2.73,0.3) {\dots};	
	\draw[thick] (3.0,0) rectangle (3.6,0.6);
	\draw[thick,->] (0,0) -- (0,3.0) node[midway,above=0.1cm,sloped] {volume};
	\draw[thick,->] (0,0) -- (4,0) node[midway,below] {\vphantom{/}scenarios};	
	\draw[thick,<-] (0.9,1.1) -- (0.9,1.6);
	\node at (1.5,1.9) {$\beta i$-th scenarios};
\end{tikzpicture}
\caption{Volume taken in each\\ scenario in the optimal solution}
\label{fig:optimal-volume}
\end{subfigure}	
\caption{Illustration of combinatorial amortization analysis}
\label{fig:combinatorial}
\end{figure}

\subsection{Our Results and Techniques}

This paper gives an optimal $4$-approximation for the Correlated Pandora's Problem.
Our results show that the Correlated Pandora's Problem is as easy as the MSSC special case, despite having arbitrary costs and the additional challenge of deciding online when to stop in the presence of different volumes.

\medskip
\noindent {\bf LP-based Algebraic Arguments.}
We deviate from the combinatorial arguments in previous work. 
Instead, we build on the LP and randomized rounding approach by \citet{BansalBFT:SODA:2021}, who gave an alternative algorithm for MSSC that also achieves the optimal $4$-approximation.
Conceptually, the LP-based arguments appear to generalize better than the combinatorial ones, for the advantages we demonstrate below. 

Algorithms like Greedy and Weitzman's Rule pick a box in each round deterministically.
While the box may be good in some or even most scenarios, it can be bad in others.
Therefore, the combinatorial arguments must amortize across different scenarios by defining the amortized objective of each scenario and using different orders in \Cref{fig:optimal-opening,fig:alg,fig:optimal-volume} before making a comparison.
By contrast, the LP-based randomized rounding approach gives a stronger guarantee of being competitive in \emph{every scenario} against the LP benchmark.
In other words, the rounding algorithm's randomness has already amortized across the scenarios, taking care of each of them in expectation. 

Freed from the burden of amortizing across scenarios, we can focus on the stopping decision and amortization within each scenario.
To compare the algorithm's objective with the LP benchmark, we bound it by a linear expression of the LP variables.
The actual objective is non-linear.
The series of manipulations and inequalities that linearize it can be seen as an amortization across the variables for the given scenario.
For example, one of these steps sums up the expression for box-opening costs and a part of the expression for the taken volume, and then applies Jensen's inequality. 
It would be difficult to replicate this kind of algebraic manipulation in a combinatorial argument.

\medskip
\noindent {\bf Poisson Rounding and Poisson Time Horizon.~}
Another ingredient of our approach is the Poisson Rounding algorithm and the concept of Poisson time horizon.
We define a Poisson process of box arrivals and design the arrival rates based on the LP variables.
The expression of the arrival rates is inspired by the $\alpha$-Point Rounding for MSSC by \citet{BansalBFT:SODA:2021}.
Then, we open the boxes in the order of their first arrivals in the Poisson time horizon. %

By replacing $\alpha$-Point Rounding with Poisson Rounding, we ensure that many events regarding the rounding and stopping decisions are independent.
Importantly, it allows us to design the stopping algorithm and analyze the objective \emph{in the Poisson time horizon}. 
Consider the event that the algorithm opens the boxes that arrive before time $\tau$ in the Poisson time horizon, and then takes box $i$.
While the actual objective depends on the opening costs of the boxes that arrive before $\tau$, the costs for opening these boxes average out over their Poisson arrivals.

\medskip
\noindent {\bf Balanced Stopping.~}
Finally, we design a Balanced Stopping algorithm, which stops when the Poisson time is at least half the minimum volume from an opened box.
Intuitively, it balances the ``exploration cost'' of opening other boxes and the ``exploitation cost'' of the taken box. 
We derive the algorithm from first principles and based on the design of Poisson Rounding.
Informally, we argue that if an opened box $i$ and the current time $\tau$ satisfy the criteria of Balanced Stopping, then one should take box $i$ even if it was not the correct box the benchmark takes, because waiting for the correct box's arrival in Poisson Rounding would be costlier than taking box $i$ now.
See \Cref{sec:stopping-dicusssion} for a brief discussion and \Cref{sec:general-cost} for details.

We have already discussed the previous results most related to our work. 
\Cref{tab:comparison} summarizes the comparison of our algorithm with the existing ones. 
We defer the discussion on further related works to \Cref{sec:related}, and refer readers to the survey by \citet{BeyhaghiC:SIGecom:2024}.

\begin{table}[th]
\renewcommand{\thefootnote}{\fnsymbol{footnote}}
\renewcommand{\arraystretch}{1.2}
\centering
\caption{Comparison of algorithms and results}
\label{tab:comparison}
\begin{tabular}{ll@{\hskip 0.3in}ll}
	\toprule
	\textbf{Algorithm} & \textbf{Order of Boxes} & \textbf{Stopping Rule} & \textbf{Ratio} \\
	\midrule
	\textbf{This work}\,\footnotemark[1] & Poisson rounding & Balanced stopping & $4$ \\
	\citet{BansalBFT:SODA:2021}\,\footnotemark[1] & $\alpha$-Point rounding & \multirow{2}{*}{Stop at $0$ volume} & \multirow{2}{*}{$4$ (MSSC)} \\
	\citet{FeigeLT:Algorithmica:2004}\,\footnotemark[2] & Greedy \\	
	\citet{GergatsouliT:NeurIPS:2023}\,\footnotemark[2] & \multicolumn{2}{l}{Weitzman's rule for ordering and stoppping} & $4.43$ \\
	\citet{ChawlaGTTZ:FOCS:2020}\,\footnotemark[2] & Greedy\,\footnotemark[3] & Ski Rental stopping & $9.22$ \\
	\bottomrule
	\multicolumn{4}{l}{\small\footnotemark[1] LP-based algebraic argument} \\[-.5ex]
	\multicolumn{4}{l}{\small\footnotemark[2] Combinatorial argument} \\[-.5ex]
	\multicolumn{4}{l}{\small\footnotemark[3] Greedy for an MSSC instance obtained by rounding box-opening costs to either $0$ or $\infty$} 
\end{tabular}
\end{table}

\section{Preliminaries}
\label{sec:prelim}

We write $[n]$ for $\{1, 2, \dots, n\}$ with the convention that $[0]$ is the empty set.
Let $x_+$ denote the function $\max \{ x, 0 \}$.
\subsection{Correlated Pandora's Problem}
\label{sec:prel}
Consider a set of boxes $[n]$. 
Each box $i$ is associated with a volume $v_i \ge 0$ and a cost $c_i \ge 0$ for opening it. 
The costs are deterministic. 
The volume vector $v \in [0, \infty)^n$, which we refer to as the \emph{scenario}, is drawn from a distribution $\Distribution$.
If $\Distribution$ is a product distribution, this corresponds to the classical Pandora's Problem by \citet{Weitzman/1979/original}.
We consider the Correlated Pandora's Problem where $\Distribution$ is a general distribution.

Both the costs and the distribution $\Distribution$ are known to the algorithm, but the realization of $v_i$ is only revealed when the algorithm opens box $i$.
In each round when at least one box remains unopened, the online algorithm either stops and takes an opened box with the minimum volume, or continues to open another box. 
If all boxes have been opened, the algorithm takes the minimum volume.
We want to minimize the total cost used to open boxes plus the volume of the taken box. 

We remark that the original model by \citet{Weitzman/1979/original} considered maximizing the taken \emph{value} $v_i$ minus the box-opening costs.
While the two versions are essentially the same for exact algorithms such as Weitzman's Rule, the maximization version is much harder for approximation algorithms, as its objective is the difference between two positive components. 
\citet{ChawlaGTTZ:FOCS:2020} showed that the maximization version of Correlated Pandora's Problem has no bounded approximation.

\medskip
\noindent
{\bf Interpretation as Time Cost.~}
It is helpful to interpret the costs and volumes as time usage by the algorithm, such as in the motivating example of LLM training. 
Hence, we sometimes say that opening box $i$ takes time $c_i$, and spending time $v_i$ to take box $i$.
We will use the terminologies in the original definition and this interpretation interchangeably.

For each box $i$, let $t_i$ be the moment when the algorithm starts to open box $i$.
Then, we say that box $i$ is \emph{closed} from time $0$ to time $t_i$, \emph{being opened} from time $t_i$ to time $t_i + c_i$, and \emph{opened} after time $t_i$.
We define these three time intervals to be left-closed and right-open. 
Finally, we refer to $t_i$ as the \textit{start time} and $t_i+c_i$ as the \textit{completion time} of box $i$.
See \Cref{fig:status of a box} for an illustration.

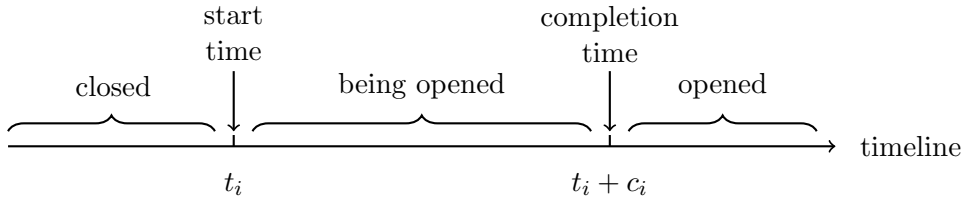
\begin{figure}[ht]
    \centering
    \begin{tikzpicture}[scale = 1]
        \draw [thick,->] (2,0) -- (13,0);
        \draw [thick,->] (5,1) -- (5,0.2);
        \draw [thick,->] (10,1) -- (10,0.2);
        \draw [thick,-] (5,0.15) -- (5,0);
        \draw [thick,-] (10,0.15) -- (10,0);
        \draw [thick,decorate,decoration={brace,amplitude=6pt,raise=0pt}] (2,0.2) -- (4.75,0.2);
        \draw [thick,decorate,decoration={brace,amplitude=6pt,raise=0pt}] (5.25,0.2) -- (9.75,0.2);
        \draw [thick,decorate,decoration={brace,amplitude=6pt,raise=0pt}] (10.25,0.2) -- (12.75,0.2);
        \node[align=center] at (5,-0.5) {$t_i$};
        \node[align=center] at (10,-0.5) {$t_i + c_i$};
        \node[align=center] at (3.4,0.8) {closed};
        \node[align=center] at (7.5,0.8) {being opened};
        \node[align=center] at (11.5,0.8) {opened};
        \node[align=center] at (5,1.5) {start \\ time};
        \node[align=center] at (10,1.5) {completion \\ time};
        \node[align=center] at (14,0) {timeline};
    \end{tikzpicture}
\caption{Timeline of box $i$}
\label{fig:status of a box} 
\end{figure}
\subsection{Partially Adaptive Algorithms}
The literature has considered three families of online algorithms: non-adaptive algorithms, fully adaptive algorithms, and partially adaptive algorithms.

A \emph{non-adaptive} algorithm opens a fixed subset of boxes in all cases, and then takes the opened box with the minimum volume.
Such algorithms are very restrictive. 

A \emph{fully adaptive} algorithm is a general online algorithm that adaptively decides whether to stop, and if not, which box to open next, based on all revealed information.
However, finding a good fully adaptive algorithm is difficult.
Computationally, \citet{DBLP:conf/approx/0001GMT23} showed that finding a constant approximation is equivalent to the open question of whether polynomial-time constant-approximation algorithms exist for the Uniform Decision Tree problem \citep[e.g.,][]{ChakaravarthyPRAM:PODS:2007,LiLM:SODA:2020}.
From a learning perspective, \citet{ChawlaGTTZ:FOCS:2020} pointed out that a fully adaptive algorithm may use the volume of the first box to decode which subsequent box has zero volume and to open that box directly, yet such correlation is impossible to learn from samples.

Following the works by \citet{ChawlaGTTZ:FOCS:2020} and \citet{GergatsouliT:NeurIPS:2023}, we will focus on the  \emph{partially adaptive} algorithms throughout this paper.
Such an algorithm fixes an order of the boxes at the beginning, and opens them one by one.
After opening each box, it adaptively decides whether to stop or open the next box.
In other words, a partially adaptive algorithm opens boxes in a non-adaptive order, but makes adaptive stopping decisions.

We remark that partially adaptive and fully adaptive algorithms are equivalent in MSSC, as the algorithm takes the first zero-volume box.
When the algorithm needs to choose a box to open, all revealed volumes must be infinite, offering no basis for an adaptive choice of the next box.

A polynomial-time algorithm is an $F$-approximation if, for any instance of the Correlated Pandora's Problem, its expected objective is at most $F$ times that of the optimal (possibly exponential time) partially adaptive algorithm. 

\section{Unit-Cost Case with Clairvoyant Stopping}
\label{sec:unit-cost}
As an overview, this section focuses on the unit-cost case, i.e., when $c_i = 1$ for all boxes $i \in [n]$, and assumes that the strategy can stop clairvoyantly at the optimal time given the order of boxes. 
Studying this special case already allows us to introduce most concepts and notations we will use later, as well as some of our new ideas that lead to the optimal $4$-approximation.
To simplify exposition, we will also assume in this section that the volumes $v_i$ are \emph{positive even integers}, to avoid floor and ceiling operations.
The final analysis in \Cref{sec:general-cost} will cover all possible values of costs and volumes.

\subsection{Linear Program Relaxation}
\label{subsec:unit-cost-lp}

By the unit-cost assumption, we consider discrete time steps, where time step $t$ corresponds to interval $[t-1,t)$.
Consider any partially adaptive algorithm.
Let $x_i(t)$ be the indicator that box $i$ is being opened during time step $t$.
For each scenario $v$, let $z_i(t \,|\, v)$ be the indicator that box $i$ is being opened during time step $t$, and is eventually taken by the algorithm. 
Relaxing the indicators to be fractions between $0$ and $1$, we get the following linear programming (LP) relaxation, which we call the \emph{Unit-Cost LP}:
\begin{align}
	\minimize \quad & \E_{v \sim \Distribution} \, \sum_{t \in [n]} \sum_{i \in [n]} \, (t + v_i) \cdot z_i(t \,|\, v)  \nonumber \tag{\textsc{Unit-Cost LP}} \\
	\subjectto \quad & \sum_{i \in [n]} x_i(t) \le 1 && \forall t \in [n]  \label{eq:unit-cost-1}\\
	& \sum_{t \in [n]} \sum_{i \in [n]} z_i(t \,|\, v) = 1 && \forall v \label{eq:unit-cost-2}\\[.5ex]
	& 0 \le z_i(t \,|\, v) \le x_i(t) && \forall i \in [n], \forall t \in [n], \forall v  \label{eq:unit-cost-3}
\end{align}

The first constraint \eqref{eq:unit-cost-1} says that the algorithm can open at most one box at each step.
The second constraint \eqref{eq:unit-cost-2} says that the algorithm must take a box in every scenario $v$.  
The third constraint \eqref{eq:unit-cost-3}  says that the event of $z_i(t \,|\, v)$ is a subset of the event of $x_i(t)$.
Finally, we only consider time steps from $1$ to $n$ by the unit-cost assumption. 
This is almost the same as the LP used by \citet{ChawlaGTTZ:FOCS:2020}, except that we omit a constraint saying that each box can be opened at most once (i.e., $\sum_t x_i(t) \leq 1$) as it is redundant for our analysis.

While this LP is conceptually easier to work with, it has uncountably many variables $z_i(t \,|\, v)$ when the support of scenarios is continuous.
We will next introduce an {\em equivalent} convex program with only $x_i(t)$ as its variables, and show how to solve it efficiently.

\medskip
\noindent {\bf Optimal $z$ Variables.~}
Observe that for any fixed scenario $v$, the best choice of $z_i(t \,|\, v)$ is entirely determined by the $x_i(t)$ variables.
In particular, given $x_i(t)$, the optimal $z_i(t \,|\, v)$ is obtained by sorting box-time pairs $(i, t)$ by the ascending order of $t + v_i$, and setting each $z_i(t \,|\, v) \le x_i(t)$ as large as possible until the second constraint is satisfied, i.e., until the sum of $z_i(t \,|\, v)$ equals $1$.

For any fixed scenario $v$, the process above reduces to finding the smallest positive integer threshold $t(v)$ such that:
\[
	\sum_{i \in [n]} \sum_{t \le t(v) - v_i} x_i(t) \ge 1
	~.
\]

This threshold $t(v)$ can also be viewed as the largest value of $t + v_i$ for which some $z_i(t \,|\, v)$ is non-zero. 
More precisely, we have that $z_i(t \,|\, v) = x_i(t)$ when $t + v_i < t(v)$, $0$ when $t + v_i > t(v)$, and some value in-between when $t + v_i = t(v)$.
\Cref{fig:definition-t-z} presents an illustration.

\begin{figure}[h]
    \centering
    \begin{tikzpicture}
        \draw[fill=blue!25] (0,0) rectangle (2,-0.5);
        \draw[thick,<->] (0,-.25) -- (2,-.25) node[pos=.5,fill=blue!25] {$v_1$};
        \draw[dotted,thick] (6,0) -- (6,-0.5);
        \draw[thick,<->] (2,-.25) -- (6,-.25) node[pos=.5,fill=white] {$t$};
        \draw[fill=blue!25] (0,-0.7) rectangle (3,-1.2);
        \draw[thick,<->] (0,-0.95) -- (3,-0.95) node[pos=.5,fill=blue!25] {$v_2$};
        \draw[fill=blue!25] (0,-1.4) rectangle (5,-1.9);
        \draw[thick,<->] (0,-1.65) -- (5,-1.65) node[pos=.5,fill=blue!25] {$v_3$};
        \draw[fill=blue!25] (0,-3.1) rectangle (4,-3.6);
        \draw[thick,<->] (0,-3.35) -- (4,-3.35) node[pos=.5,fill=blue!25] {$v_n$};
        \draw[thick] (0,0.4) -- (0,-4);
        \draw (0,0) -- (13,0);
        \node at (-1,-.25) {Box $1$};
        \draw (0,-0.5) -- (13,-0.5);
        \draw (0,-0.7) -- (13,-0.7);
        \node at (-1,-0.95) {Box $2$};
        \draw (0,-1.2) -- (13,-1.2);
        \draw (0,-1.4) -- (13,-1.4);
        \node at (-1,-1.65) {Box $3$};
        \draw (0,-1.9) -- (13,-1.9);
        \node at (-1,-2.5) {\dots};
        \node at (5,-2.5) {\dots};
        \draw (0,-3.1) -- (13,-3.1);
        \node at (-1,-3.35) {Box $n$};
        \draw (0,-3.6) -- (13,-3.6);
        \draw[dashed] (9.1,0.2) -- (9.1,-3.85);
        \node at (8.6,-4.2) {$t(v)-1$};        
        \draw[dashed] (9.8,0.2) -- (9.8,-3.85);
        \node at (10, -4.2) {$t(v)$};        
        \draw[thick,->] (6,0.6) -- (6,0.1);
        \node at (6,.9) {$z_1(t) = x_1(t)$};
        \draw[thick,->] (9.45,0.6) -- (9.45,0.1);
        \node at (9.45,.9) {$0 \le z_1(t) \le x_1(t)$};        
        \draw[thick,->] (12,0.6) -- (12,0.1);
        \node at (12.5,.9) {$z_1(t) = 0$};        
    \end{tikzpicture}
    \caption{Illustration of the threshold $t(v)$ and the optimal $z_i(t \,|\, v)$ for scenario $v$ given $x_i(t)$}
    \label{fig:definition-t-z}
\end{figure}
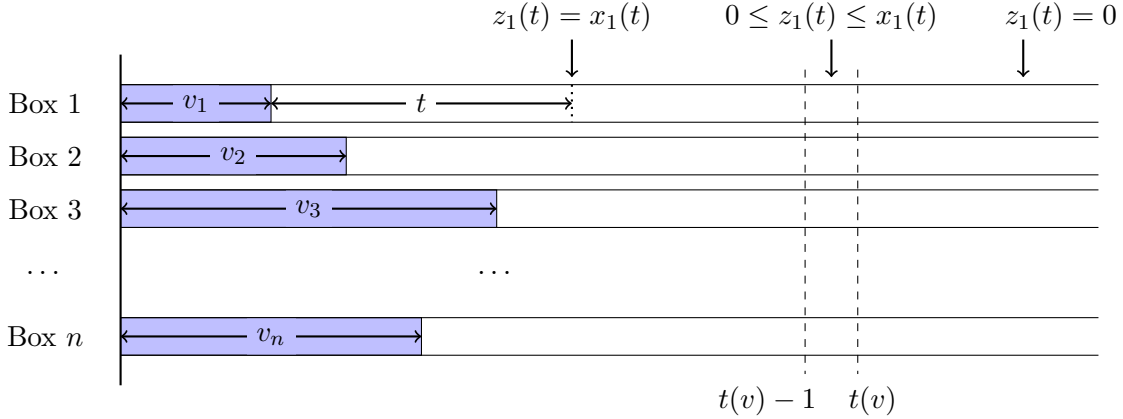

\medskip
\noindent {\bf Probabilistic Interpretation.}
Recall that for any scenario $v$, the variables $z_i(t \,|\, v)$ sum to $1$.
Treating $z_i(t \,|\, v)$ as the probability mass of a distribution over the box-time pairs, the LP objective for scenario $v$ equals the expectation of integer random variable $t + v_i$, which equals the sum of its complementary CDF:
\[
	\sum_{t \in [n]} \sum_{i \in [n]} (t + v_i) \cdot z_i(t \,|\, v) ~=~ \sum_{t = 1}^{t(v)} \bigg( 1 - \sum_{i \in [n]} \sum_{t' < t - v_i} z_i(t' \,|\, v) \bigg)
	~.
\]

By the relation between $z_i(t \,|\, v)$ and $x_i(t)$, this equals $\sum_{t = 1}^\infty \big( 1 - \sum_{i \in [n]} \sum_{t' < t - v_i} x_i(t') \big)_+$; recall that $x_+ = \max \{x, 0\}$. 
Putting together, we have the following \emph{Unit-Cost Convex Program}:%
\footnote{Even though the $\max(\cdot,0)$ in the objective can be cast as an LP, this would require introducing a variable for each scenario. To avoid this, and to distinguish from $\textsc{Unit-cost LP}$, we refer to this as a convex program.}
\begin{align*}
	\minimize \quad & \E_{v \sim \Distribution} \, \sum_{t = 1}^\infty \bigg( 1 - \sum_{i \in [n]} \sum_{t'< t - v_i} x_i(t') \bigg)_+ \tag{\textsc{Unit-Cost CP}} \\
	\subjectto \quad & \sum_{i \in [n]} x_i(t) \le 1 && \forall t \in [n] \\
	& x_i(t) \ge 0 && \forall i \in [n], \forall t \in [n]
\end{align*}

Let $\CP_{\textsc{Unit}}$ denote the optimal objective value of this convex program.
Let $\OPT$ denote the optimal expected cost achievable by partially adaptive policies.
By the discussion above, we have:

\begin{lemma}
    \label{lem:unit-cost-lp-relaxation}
    For any unit-cost instance of Correlated Pandora's Problem, we have $\CP_{\textsc{Unit}} \le \OPT$.
\end{lemma}

If the distribution over scenarios has a polynomial-size support and is given explicitly, then \textsc{Unit-Cost CP} is of polynomial size and can be solved efficiently.
Moreover, it can be solved almost optimally in general. 
Observe that (i) the variables and the polytope of the feasible region are independent of the scenarios and their distribution, and (ii) the objective is linearly separable for different scenarios.
These two observations allow us to solve the convex program approximately using stochastic gradient descent in polynomial time and using a polynomial number of scenarios sampled from distribution $\Distribution$.

\begin{lemma}
	\label{lem:unit-cost-lp-solvability}
	For any $ \varepsilon>0$, we can compute with high probability a feasible solution to \textsc{Unit-Cost CP} with objective at most $(1 + \varepsilon)\,\CP_{\textsc{Unit}}$, with time and sample complexity polynomial in $n$ and $\frac{1}{\varepsilon}$.
\end{lemma}

We view \textsc{Unit-Cost CP} as an efficient way to compute an (approximately) optimal solution $x_i(t)$, and implicitly $z_i(t \,|\, v)$. 
Although our algorithm does not use $z_i(t \,|\, v)$, our analysis will use these variables.

\subsection{(Discrete-Time) Poisson Rounding}
\label{subsec:unit-cost-rounding}

Recall that $x_i(t)$ is the probability that box $i$ is being opened at time step $t$.
We further let:
\[
	\bar{x}_i(t) ~\defeq~ \frac{1}{t} \sum_{t' = 1}^{\min\{t, n\}} x_i(t') 
\]
denote the ``average probability'' that box $i$ is being opened during the first $t$ time steps.
Since $\sum_{i \in [n]} x_i(t) \le 1$ at each time $t$, we also have that $\sum_{i \in [n]} \bar{x}_i(t) \le 1$ for each $t$.
Hence, for any time step $t$, we may treat $\bar{x}_i(t)$ as specifying a distribution over the boxes $i$. 
Consider the following rounding algorithm.

\begin{algorithm}{(Discrete-Time) Poisson Rounding}

\smallskip
For time steps $\tau \ge 1$:
\begin{enumerate}
    \item Sample a box, where box $i \in [n]$ is sampled with probability $\bar{x}_i(\lceil \frac{\tau}{2} \rceil)$.%
	\footnote{We sample a dummy box with probability $1 - \sum_i \bar{x}_i(\lceil \frac{\tau}{2} \rceil)$, which represents not opening any box.}
    \item Open the sampled box.%
    \footnote{Note that we allow opening a box multiple times. 
    The analysis holds even with such redundant costs.\label{fn:multiple-open}}
\end{enumerate}
\end{algorithm}

We will consider the time steps of Poisson Rounding as happening in a different time horizon, in contrast to the real time horizon.
We will refer to the time horizon of Poisson Rounding as the \emph{Poisson time horizon} and use Greek letters such as $\tau$ for time steps there, and the original one as the \emph{real time horizon} and use English letters such as $t$ for those time steps.
While this treatment is superfluous in the unit-cost case as these two time horizons are identical, it will be crucial in the general case, so we hope to foreshadow the relevant concepts. 

\begin{lemma}
	\label{lem:unit-cost-two-horizons}
	We spend at most $\tau$ steps in the real time horizon to open the boxes sampled in the first $\tau$ steps in the Poisson time horizon.
\end{lemma}

For any box $i \in [n]$, we define its \emph{arrival time} $\alpha_i$ as the earliest time step when box $i$ is sampled in the Poisson time horizon.
Clearly, by \Cref{lem:unit-cost-two-horizons}, we get that: 

\begin{lemma}
	\label{lem:unit-cost-completion-time}
	For any box $i \in [n]$, it is opened by time step $\alpha_i$ in the real time horizon.
\end{lemma}

In the next subsection, we will design criteria for stopping and taking the best opened box.
To do so, we will need to bound the probability that the boxes arrive later than some respective thresholds, and as a result, the algorithm pays a high cost.
The next lemma serves this purpose.

\begin{lemma}
	\label{lem:unit-cost-poisson-bound}
	For any positive integer thresholds $(\theta_i)_{i \in [n]}$, we have:
	\[
		\Pr \big[\, \forall i \in [n], \alpha_i > 2 \theta_i \,\big] 
        ~\le~ 
        \exp \Big( - 2\sum_{i \in [n]} \sum_{\tau \in [\theta_i]} \bar{x}_i(\tau) \Big)
		~.
	\]		
\end{lemma}

\begin{proof}
    By definition, this probability equals  $\prod_{\tau = 1}^\infty \big( 1 - \sum_{i : \tau \le \theta_i} \bar{x}_i(\tau) \big)^2$, which is upper bounded by the right-hand-side because $1-x \le e^{-x}$.
\end{proof}

\subsection{Clairvoyant Stopping}
\label{subsec:unit-cost-clairvoyant}

In this subsection, we ignore the online aspect of when to stop, and consider the following {\em Clairvoyant Stopping}, which picks the best box as soon as it is opened. 

\begin{algorithm}{Clairvoyant Stopping with Poisson Rounding}
\begin{enumerate}
    \item Sample arrival time $\alpha_i$ of boxes $i \in [n]$ using  (Discrete-Time) Poisson Clock Rounding.
    \item Find box $i^* = \arg\min_{i \in [n]} ~ \alpha_i + v_i$.
    \item Open boxes in ascending order of $\alpha_i$ and take box $i^*$ as soon as it is opened.
\end{enumerate}
\end{algorithm}

\medskip

\begin{theorem}
	\label{thm:unit-cost-clairvoyant}
	Clairvoyant Stopping with Poisson Rounding is $4$-competitive.
\end{theorem}

\begin{proof}
	It suffices to show that for any scenario $v$, the expected objective of the algorithm is at most $4$ times the LP objective for scenario $v$. Fix a scenario $v$. We will omit the dependence on $v$ and write $z_i(t)$ for $z_i(t \,|\, v)$ for brevity.
	The LP objective is then:
	\begin{equation}
		\label{eqn:unit-cost-lp-objective}	
		\sum_{i \in [n]} \sum_{t \in [n]} \big( t + v_i \big) \cdot z_i(t)
		~.
	\end{equation}
	
    By definition, the algorithm's objective is $\ALG = \min_{i \in [n]} \alpha_i + v_i$.
    Recall that we have assumed the volumes $v_i$ to be even in this section.
	For any integer threshold $\theta$, the probability that the algorithm pays at least $2\theta$ is:
	\begin{equation}
    \label{eq:ov-alg-cost}
		\Pr \big[\, \ALG \ge 2\theta \,\big] 
		~=~ 
		\Pr \Big[\, \forall i, \alpha_i \ge \big(2\theta - v_i\big)_+ \,\Big]		
        \underset{(\text{Lem.}\, \ref{lem:unit-cost-poisson-bound})}{\leq}
		\exp \bigg( - 2\sum_{i \in [n]} \sum_{\tau < (\theta - \frac{v_i}{2})_+} \bar{x}_i(\tau) \bigg)
		~.
	\end{equation}
    
    As the algorithm's expected objective can be written as $\sum_{\theta = 1}^\infty \Pr \big[ \ALG \ge \theta \big]$, and the LP objective \eqref{eqn:unit-cost-lp-objective} is linear in $z_i(t)$, we will relax and linearize the right-hand-side of \eqref{eq:ov-alg-cost} for easy comparison.
	First, we rewrite the right-hand-side by expressing $\bar{x}_i(\tau)$ in terms of $x_i(\tau)$ as follows:
	\[
		 \exp \bigg( - 2 \sum_{i \in [n]} \sum_{\tau \in [n]}  x_i(\tau) \sum_{\tau \le \tau' < (\theta - \frac{v_i}{2})_+} \frac{1}{\tau'} \bigg)
		 ~.
	\]
    
	Using that  $\frac{1}{\tau'} \ge \ln \frac{\tau'+1}{\tau'}$, the inner sum telescopes, and gives that the above is at most:
	\[
		\exp \bigg( - 2 \sum_{i \in [n]} \sum_{\tau \in [n]} x_i(\tau) \cdot \ln \frac{\max \{ (\theta - \frac{v_i}{2})_+, \tau \}}{\tau} \bigg)
		~,
	\]
	where the maximization operation ensures that the coefficient for $x_i(\tau)$ is zero for $\tau \ge (\theta - \frac{v_i}{2})_+$. 
	
	By the LP constraint $z_i(\tau) \le x_i(\tau)$, we further relax the above by replacing $x_i(\tau)$ with $z_i(\tau)$:
	\[
		\exp \bigg( - 2 \sum_{i \in [n]} \sum_{\tau \in [n]} z_i(\tau) \cdot \ln \frac{\max \{ (\theta - \frac{v_i}{2})_+, \tau \}}{\tau} \bigg)
		~.
	\]	

	Now, by the LP constraint $\sum_{i \in [n]} \sum_{t \in [n]} z_i(t) = 1$ and Jensen's inequality, this is at most:
	\[
		\sum_{i \in [n]} \sum_{\tau \in [n]} z_i(\tau) \cdot \Big( \frac{\tau}{\max \{(\theta - \frac{v_i}{2})_+, \tau\}} \Big)^2
		~.
	\]        
	
	Putting it all together, the expected objective of the algorithm is:
	\[
		2 \sum_{\theta = 1}^\infty \Pr \big[ \ALG \ge 2\theta \big] ~\le~ 2 \sum_{\theta = 1}^\infty \sum_{i \in [n]} \sum_{t \in [n]} z_i(t) \cdot \Big( \frac{t}{\max \{(\theta - \frac{v_i}{2})_+, t\}} \Big)^2
		~.
	\]
	
	For any box $i \in [n]$ and any time $t \in [n]$, the coefficient of $z_i(t)$ on the right-hand-side is:
    \[
        2 \sum_{\theta = 1}^\infty \Big( \frac{t}{\max \{(\theta - \frac{v_i}{2})_+, t\}} \Big)^2
        = 2 \, \Bigg( t + \frac{v_i}{2}  + \sum_{\theta > t + \frac{v_i}{2}} \Big( \frac{t}{\theta - \frac{v_i}{2}} \Big)^2 \Bigg)
        = 2t + v_i + 2t^2 \sum_{i = t+1}^\infty \frac{1}{i^2}
        ~.
    \]
    
    Here the first step follows by noting that the term on the left is $1$ for $\theta \leq t+v_i/2$.

	By the well-known relaxation that  $\sum_{i=t+1}^\infty \frac{1}{i^2} \le \sum_{i=t+1}^\infty \frac{1}{i(i-1)} = \frac{1}{t}$, the above coefficient is at most $4t + v_i$, which is at most $4$ times the LP coefficient $t + v_i$ of $z_i(t)$ in \Cref{eqn:unit-cost-lp-objective}.
\end{proof}

\begin{remark} 
	\label{remark:unit-cost-clairvoyant}
    Notice that	Clairvoyant Stopping with Poisson Rounding remains $4$-competitive even if we get up to $k \le 4$ times the actual volumes, i.e., $k v_i$ instead of $v_i$ from any box $i$, where we redefine $i^* = \argmin_{i \in [n]} \alpha_i + kv_i$ accordingly.
	In that case, the coefficient of $z_i(t \,|\, v)$ at the end of the proof increases from $4t + v_i$ to  $4t + kv_i$, still at most $4$ times the LP coefficient $t + v_i$.
\end{remark}

\subsection{Brief Discussion on Stopping Algorithms}
\label{sec:stopping-dicusssion}

We conclude the section with a discussion on when to stop and take the best opened box. 
This decision is trivial in the special case of MSSC, where the volumes are either zero or infinite.
In this sense, the design of the stopping algorithm is the main obstacle to extending the $4$-approximation ratio to the more general Correlated Pandora's Problem.
Below we discuss three approaches.
All of them are related to the Ski Rental problem from some perspectives, but finding the final solution requires examining the Correlated Pandora's Problem and our Poisson Rounding algorithm from first principles.

\medskip
\noindent 
{\bf First Attempt: Generalized Ski Rental.~}
The Ski Rental problem proceeds in $T$ steps, but $T$ is unknown to the algorithm.
At each step until the algorithm buys a pair of skis, it needs to decide whether to \emph{rent} at a cost of $1$ or to \emph{buy} at a cost of $B$.
We want to minimize the total cost. 

\citet{ChawlaGTTZ:FOCS:2020} observed that the stopping problem of Correlated Pandora's Problem is related to Ski Rental.
Consider a sequence of boxes such that opening the first box costs $0$, and opening each of the other boxes costs $1$.
For any $T$, consider a corresponding scenario $T$, in which the first box's volume is $B$, the next $T-1$ boxes' volumes are infinite, and the remaining boxes' volumes are $0$.
After the first box is opened for free, opening each of the next $T$ boxes corresponds to renting for another day, while stopping and taking volume $B$ corresponds to buying.
 
By generalizing the optimal $\frac{e}{e-1}$-competitive Ski Rental algorithm, \citet{ChawlaGTTZ:FOCS:2020} showed an approximation ratio of $(3+2\sqrt{2}) \cdot \frac{e}{e-1}$, where $3+2\sqrt{2}$ is the ratio of clairvoyant stopping for their order of boxes.
Although we have improved the $3+2\sqrt{2}$ ratio to $4$ using Poisson Rounding, this approach at best gives a $\frac{4e}{e-1}$-approximation.

\medskip
\noindent 
{\bf Second Attempt: Copycat Clairvoyant Stopping.~}
Given the success of Clairvoyant Stopping, and the slack in the guarantee with respect to volumes (\Cref{remark:unit-cost-clairvoyant}), our next attempt is to copy the clairvoyant decision at a cost proportional to its volume, via a ski-rental type of tradeoff.
More precisely, we propose the following \emph{Delayed Activation} algorithm:

\begin{enumerate}
	\item When a box $i$ is opened at time $\alpha_i$, mark $\alpha_i + v_i$ as the time when box $i$ can be taken. 
	\item Stop at the earliest marked time and take the corresponding box $i' = \argmin_i\, \alpha_i + v_i$.
\end{enumerate}

Delayed Activation spends $\alpha_{i'} + v_{i'}$ time to open boxes, before it takes box $i'$ with volume $v_{i'}$, yielding an objective of $\alpha_{i'} + 2v_{i'}$. Does such a delayed activation imply a $4$-approximation using \Cref{remark:unit-cost-clairvoyant} with $k = 2$?
Unfortunately and perhaps surprisingly, the answer is negative --- Delayed Activation gets objective $\alpha_{i'} + 2v_{i'}$ for the box $i' = \argmin_i \alpha_i + v_i$, but \Cref{remark:unit-cost-clairvoyant} with $k = 2$ holds for box $i^* = \argmin_i \alpha_i + 2v_i$ instead.
In the worst case, $\alpha_{i'} + 2v_{i'}$ could be much larger than $\alpha_{i^*} + 2v_{i^*}$.
For instance, consider $\alpha_{i'} = 0$, $v_{i'} = 1-\varepsilon$, $\alpha_{i^*} = 1$, and $v_{i^*} = 0$.
Delayed Activation would stop at time $1 - \varepsilon$, right before box $i^*$ arrives, and pay almost twice in the objective because $\alpha_{i'} + 2v_{i'} \approx 2$ while $\alpha_{i^*} + 2v_{i^*} = 1$.

Although the above direct application of \Cref{remark:unit-cost-clairvoyant} fails, we find an alternative factor-revealing LP approach to utilize the remark, and obtain the following result.\footnote{See Theorem \ref{thm:delayed activation} in \Cref{app:delayed-activation} for a more precise version of the theorem.}

\begin{theorem}
    A randomized version of Delayed Activation is a $4.075$-approximation.
\end{theorem}

Surprisingly and unfortunately, the limit of this approach still falls short of the optimal $4$-approximation by a small gap of $0.075$.
Nonetheless, it already improves the state-of-the-art ratio of $4.43$ by \citet{GergatsouliT:NeurIPS:2023}. 
Since we already have a $4$-approximation using a different approach, we defer the formal definition of Delayed Activation and its analysis to \Cref{app:delayed-activation}.

Delayed Activation can be seen as a generalization of a $2$-competitive deterministic Ski Rental algorithm, which rents for up to $B$ days and then buys.
In the above interpretation of Ski Rental as an instance of Correlated Pandora's Problem, the first box with volume $v_i = B$ is opened at time $\alpha_i = 0$.
Hence, it will be taken by Delayed Activation after step $\alpha_i  + v_i = B$, unless we get a zero-volume box before that, which corresponds to $T \le B$.

\medskip
\noindent 
{\bf Final Solution: Balanced Stopping.~}
To get an optimal $4$-approximation algorithm, we design another algorithm called \emph{Balanced Stopping}, which takes an opened box $i$ when the time the algorithm has used to open other boxes is comparable to the time to be spent on taking box $i$.
In this sense, it balances the ``exploration cost'' of opening boxes and the ``exploitation cost'' of the box it takes. 
For the discussion in this section, we consider an inaccurate version of Balanced Stopping, omitting a complication that arises in the general case:
\begin{enumerate}
	\item Whenever a box $i$ is opened at time $\alpha_i$, mark the time $\max\{ \alpha_i, v_i \}$ as the earliest time when this box can be taken. 
	\item Stop at the first marked time and take the corresponding box $i^* = \argmin_i \max \{ \alpha_i, v_i \}$.
\end{enumerate}

To understand why $\max\{ \alpha_i, v_i \}$ is the natural time to take box $i$ under Poisson Rounding, consider the following thought experiment.
Recall that Poisson Rounding sets the arrival rates of the boxes based on variables $x_i(t)$, oblivious to the realization of the scenario.
Given any scenario $v$, we partition the box arrivals into ``good'' ones triggered by the scenario-specific variables $z_i(t \,|\, v)$, and ``bad'' ones triggered by the {\em slacks} $x_i(t) - z_i(t \,|\, v)$.
Intuitively, we want to take a ``good'' arrival and compare its cost to the LP coefficient of the corresponding variable $z_i(t \,|\, v)$.
In contrast, taking a ``bad'' arrival is risky because the slacks do not contribute to the LP objective for scenario $v$.
However, the algorithm cannot distinguish ``good'' and ``bad'' arrivals.
Instead, let us pessimistically assume the box $i$ at hand to be ``bad'', and consider whether to take it at time $t$?
Here is the key observation: from any time $t$ onwards, Poisson Rounding takes $t$ amount of time in expectation to get another ``good'' arrival.
Therefore, even if we anticipate the next ``good'' arrival to have zero volume, we would still want to take the potentially ``bad'' box $i$ if $v_i \le t$.
Finally, we take the max between $v_i$ and $\alpha_i$ as we can only take a box after its arrival time $\alpha_i$ in the real time horizon.

Interestingly, Balanced Stopping can also be seen as a generalization of the $2$-competitive deterministic Ski Rental algorithm, even though the above thought experiment is unrelated to the Ski Rental problem.
In the interpretation of Ski Rental as an instance of Correlated Pandora's Problem, we either take the first box opened at time $\alpha_i = 0$, or a later box with a volume of $0$.
Since $\max\{\alpha_i, v_i\} = \alpha_i + v_i$ when $\alpha_i = 0$ or $v_i = 0$, Balanced Stopping and Delayed Activation coincide in Ski Rental, and both degenerate to the $2$-competitive deterministic algorithm.

\section{General Case}
\label{sec:general-cost}

We now consider general instances of the Correlated Pandora's Problem.
We first describe in \Cref{subsec:general-cost-convex-program} the linear and convex programs that extend those in \Cref{sec:unit-cost} for the unit cost case.
We then describe our algorithm in Sections \ref{subsec:general-cost-rounding} and \ref{subsec:general-cost-balanced}, which consists of (i) a procedure to round the fractional LP solution, and
(ii) a stopping rule that decides when to stop opening new boxes and take the best box at hand.
Finally, in Section \ref{subsec:general-approximation}, we show that the overall algorithm achieves a $4$-approximation.

\subsection{Linear and Convex Program Relaxations}
\label{subsec:general-cost-convex-program}

Given a general instance of the problem, we consider a natural LP relaxation as follows.
It will be convenient to work in continuous time.%
\footnote{This can be easily discretized with arbitrarily small loss, see e.g., the proof of Lemma \ref{lem:general-cost-lp-solvability}.}
For each box $i$ and time $t$, variable $X_i(t)$ denotes the probability that box $i$'s start time is at most $t$, or equivalently, $X_i(t)$ is the cumulative density function (CDF) of box $i$'s start time.
Similarly, variable $Z_i(t \,|\, v)$ denotes the probability that box $i$'s start time is at most $t$, and the algorithm takes box $i$ in scenario $v$. 

\smallskip

With these variables, consider the following \emph{General-Cost Linear Program} relaxation:
\begin{align}
    \minimize \quad & \E_{v \sim \Distribution} \, \sum_{i \in [n]} \int_0^\infty (t+c_i+v_i) \dif{Z_i(t \,|\, v)} 
    \tag{\textsc{General LP}} \\
	\subjectto \quad %
        & \sum_{i \in [n]} \Big(X_i(t) - X_i \big((t-c_i)_+ \big) \Big) \leq 1 && \forall t \ge 0 
        \label{eq:general-lp-1} \\
        & \sum_{i \in [n]} Z_i(\infty \,|\,v) = 1 && \forall v \in \Distribution 
        \label{eq:general-lp-2} \\
        & 0 \leq Z_i(t' \,|\, v) - Z_i(t \,|\, v) \leq X_i(t') -X_i(t) && \forall v \in \Distribution, \forall i \in [n], \forall t \leq t'
        \label{eq:general-lp-3} \\[2ex]
        & 0 \leq Z_i(t \,|\, v), X_i(t) \leq 1 && \forall v \in \Distribution, \forall i \in [n], \forall t \ge 0 \label{eq:general-lp-4}
\end{align}

Let us explain the meaning of these constraints and see why the above linear program is a relaxation of the Correlated Pandora's Problem.

Constraint \eqref{eq:general-lp-1} ensures that at most one box is being opened at any time $t$. 
To see this, note that $X_i(t) - X_i\big((t-c_i)_+\big)$ is the probability that box $i$'s start time is between $t$ (inclusive) and $(t-c_i)_+$ (exclusive)  (as $X_i(t)$ is the CDF of box $i$'s start time), and thus is the probability that box $i$ is being opened at time $t$.

Constraint \eqref{eq:general-lp-2} says that the algorithm must take one box at the end in every scenario.

The first inequality in Constraint \eqref{eq:general-lp-3} says that
$Z_i(t \,|\, v)$ and $X_i(t)$ are non-decreasing over $t$ because they are CDFs of different events.
The second inequality says that the event captured by $Z_i(t \,|\, v)$ is a subset of the event captured by $X_i(t)$ --- during any time interval $(t,t']$, the probability of the former event is at most that of the latter.

Constraint \eqref{eq:general-lp-4} are the lower and upper bounds for probability measures. 

\medskip
\noindent
{\bf Efficiently Solvable Convex Program.~}
Next, we derive a succinct convex program with only $X_i(t)$ as its variables.
Given any $X_i(t)$, for any scenario $v$, the optimal choice of $Z_i(t \,|\, v)$ assigns probability masses to box-time pairs $(i, t)$ to minimize the expectation of $t + c_i + v_i$, subject to second and third constraints.
Similar to the unit-cost case, this reduces to finding the minimum threshold $t(v)$ such that the total probability mass of all boxes $i \in [n]$ according to $X_i(t)$, for their respective time intervals satisfying $t + c_i + v_i \le t(v)$, is at least $1$. That is,
   \[	
   \sum_{i \in [n]} X_i \big( (t(v) - c_i - v_i)_+ \big) \ge 1
   ~.
   \]

Then, for each $i$, it is optimal to let $Z_i(t \,|\, v) = X_i(t)$ for $t < t(v) - v_i - c_i$, and let $Z_i(t \,|\, v)$ be some constant that is at most $X_i\big( (t(v) - v_i - c_i)_+ \big)$ for $t \ge t(v) - v_i - c_i$, subject to the constraint across $i \in [n]$ that $\sum_{i \in [n]} Z_i(t \mid v) = 1$.
\Cref{fig:definition-t-z-general} presents an illustration.

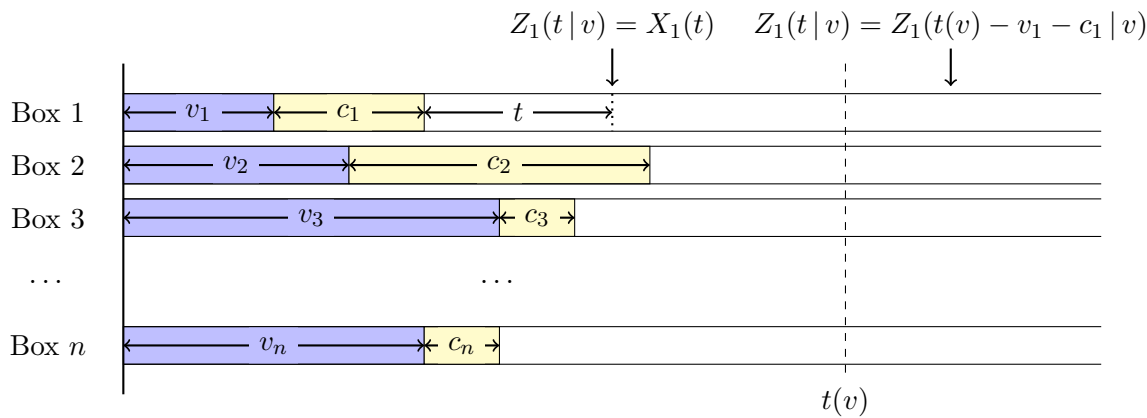
\begin{figure}[h]
    \centering
    \begin{tikzpicture}
        \draw[fill=blue!25] (0,0) rectangle (2,-0.5);
        \draw[thick,<->] (0,-.25) -- (2,-.25) node[pos=.5,fill=blue!25] {$v_1$};
        \draw[fill=yellow!25] (2,0) rectangle (4,-0.5);
        \draw[thick,<->] (2,-.25) -- (4,-.25) node[pos=.5,fill=yellow!25] {$c_1$};
        \draw[dotted,thick] (6.5,0) -- (6.5,-0.5);
        \draw[thick,<->] (4,-.25) -- (6.5,-.25) node[pos=.5,fill=white] {$t$};
        \draw[fill=blue!25] (0,-0.7) rectangle (3,-1.2);
        \draw[thick,<->] (0,-0.95) -- (3,-0.95) node[pos=.5,fill=blue!25] {$v_2$};
        \draw[fill=yellow!25] (3,-0.7) rectangle (7,-1.2);
        \draw[thick,<->] (3,-0.95) -- (7,-0.95) node[pos=.5,fill=yellow!25] {$c_2$};        
        \draw[fill=blue!25] (0,-1.4) rectangle (5,-1.9);
        \draw[thick,<->] (0,-1.65) -- (5,-1.65) node[pos=.5,fill=blue!25] {$v_3$};
        \draw[fill=yellow!25] (5,-1.4) rectangle (6,-1.9);
        \draw[thick,<->] (5,-1.65) -- (6,-1.65) node[pos=.5,fill=yellow!25] {$c_3$};    
        \draw[fill=blue!25] (0,-3.1) rectangle (4,-3.6);
        \draw[thick,<->] (0,-3.35) -- (4,-3.35) node[pos=.5,fill=blue!25] {$v_n$};
        \draw[fill=yellow!25] (4,-3.1) rectangle (5,-3.6);
        \draw[thick,<->] (4,-3.35) -- (5,-3.35) node[pos=.5,fill=yellow!25] {$c_n$};   
        \draw[thick] (0,0.4) -- (0,-4);
        \draw (0,0) -- (13,0);
        \node at (-1,-.25) {Box $1$};
        \draw (0,-0.5) -- (13,-0.5);
        \draw (0,-0.7) -- (13,-0.7);
        \node at (-1,-0.95) {Box $2$};
        \draw (0,-1.2) -- (13,-1.2);
        \draw (0,-1.4) -- (13,-1.4);
        \node at (-1,-1.65) {Box $3$};
        \draw (0,-1.9) -- (13,-1.9);
        \node at (-1,-2.5) {\dots};
        \node at (5,-2.5) {\dots};
        \draw (0,-3.1) -- (13,-3.1);
        \node at (-1,-3.35) {Box $n$};
        \draw (0,-3.6) -- (13,-3.6);
        \draw[dashed] (9.6,0.4) -- (9.6,-3.8);
        \node at (9.6, -4.1) {$t(v)$};        
        \draw[thick,->] (6.5,0.6) -- (6.5,0.1);
        \node at (6.5,.9) {$Z_1(t \,|\, v) = X_1(t)$};
        \draw[thick,->] (11,0.6) -- (11,0.1);
        \node at (11,.9) {$Z_1(t \,|\, v) = Z_1(t(v) -v_1 - c_1 \,|\, v)$};        
    \end{tikzpicture}
    \caption{Illustration of the threshold $t(v)$ and the optimal $Z_i(t\,|\,v)$ for scenario $v$ given $X_i(t)$}
    \label{fig:definition-t-z-general}
\end{figure}

Based on this relation between $X_i(t)$ and $Z_i(t \,|\, v)$, a standard computation similar to that in the unit-cost case gives:
\begin{align*}
	 \sum_{i \in [n]} \int_0^\infty (t+c_i+v_i) \dif{Z_i(t \,|\, v)} 
	 &
	 ~=~
	 \int_0^\infty \bigg( 1 - \sum_{i \in [n]} Z_i\big((t - v_i - c_i)_+ \,|\, v\big) \bigg) \dif{t} \\
	 &
	 ~=~
	 \int_0^\infty \bigg( 1 - \sum_{i \in [n]} X_i\big((t - v_i - c_i)_+\big) \bigg)_+ \dif{t} 
	 ~.
\end{align*}

Given these observations, we derive the following General-Cost Convex Program:
\begin{align}
    \minimize \quad & \E_{v \sim \Distribution} \, \int_0^\infty \bigg( 1 - \sum_{i \in [n]} X_i\big((t - v_i - c_i)_+\big) \bigg)_+ \dif{t}
    \tag{\textsc{General CP}} \\
	\subjectto \quad %
        & \sum_{i \in [n]} \Big(X_i(t) - X_i \big((t-c_i)_+ \big) \Big) \leq 1 && \forall t \ge 0 
        \label{eq:general-lp2-1} \\
        & X_i(t') - X_i(t) \le 0 && \forall t \ge t' \ge 0 \\[2ex]
        & 0 \leq X_i(t) \leq 1 &&  \forall i \in [n], \forall t \ge 0 \notag
\end{align}

Let $\CP_{\textsc{General}}$ denote the optimal objective value of this convex program. 
We have:

\begin{lemma}
    \label{lem:general-cost-lp-relaxation}
    For any instance of the Correlated Pandora's Problem, we have $\CP_{\textsc{General}} \le \OPT$.
\end{lemma}

By standard discretization and sampling, this program can be solved efficiently to an arbitrary accuracy.
We defer the proof of the following  \Cref{lem:general-cost-lp-solvability} to \Cref{app:general-cost}.

\begin{lemma}
\label{lem:general-cost-lp-solvability}
	For any $\varepsilon >0$, we can compute with high probability a feasible solution to \textsc{General CP} with objective at most $(1 + \varepsilon)\,\CP_{\textsc{General}}$, with time and sample complexity polynomial in $n$, $1/\varepsilon$, and $\frac{\max_i c_i}{\min_i c_i}$. 
\end{lemma}

\subsection{Poisson Rounding}
\label{subsec:general-cost-rounding}
We now extend the Poisson Rounding algorithm to general costs and volumes. 
Recall that in the discrete-time model for the unit-cost case, $\bar{x}_i(t)$ denotes the ``average probability'' that box $i$ is being opened in the first $t$ time steps. 
By the discussion on Constraint~\eqref{eq:general-lp-1} above, $X_i(t) - X_i\big((t-c_i)_+\big)$ captures the probability that box $i$ is being opened at time $t$. 
Thus, we define:
\begin{align*}
    \bar{x}_i(t) 
    &
    ~\defeq~ \frac{1}{t} \int_0^t \Big( X_i(t') - X_i\big((t'-c_i)_+\big) \Big)\dif{t'}
    \\
    &
    \,~=~\, \frac{1}{t} \int_0^t \min \{t-t', c_i\} \dif{X_i(t')}
    ~.
    \tag{integration by parts}
\end{align*}

Both expressions after the first and second equalities will be useful later.
By Constraint \eqref{eq:general-lp2-1}, we can still consider $\bar{x}_i(t)$ as a distribution over the boxes as:
\begin{equation}
    \label{eqn:unit-poisson-rate}
    \sum_{i \in [n]} \bar{x}_i(t) 
    = \frac{1}{t} \int_0^t \sum_{i \in [n]} \Big( X_i(t') - X_i\big((t'-c_i)_+\big) \Big)\dif{t'} 
    \leq 1
    ~.
\end{equation}

\begin{algorithm}{(Continuous-Time) Poisson Rounding}
\begin{enumerate}
    \item Each box $i \in [n]$ arrives by a non-homogeneous Poisson process with rate $\frac{1}{c_i} \bar{x}_i(\frac{\tau}{2})$.
	\item Open each box on its first arrival.%
    \footnote{The analysis holds even if we allow opening a box multiple times.}%
    \footnote{We omit ties because they are zero-measure events.}
\end{enumerate}
\end{algorithm}

The $\frac{1}{c_i}$ factor cancels the cost $c_i$ of each arrival of box $i$, so that the expected cost spent on opening a box $i$ arriving at time $\tau$ equals $\bar{x}_i(\frac{\tau}{2})$.

\medskip
\noindent {\bf Poisson and Real Time Horizons.~} 
Recall that we think of the Poisson process of the rounding algorithm as happening in a different time horizon, in contrast to the real time horizon.
We refer to the time horizon of Poisson Rounding as the \emph{Poisson time horizon}, and the original one as the \emph{real time horizon}.
We use Greek letters, such as $\tau$, for time in the former horizon, and English letters, such as $t$, for time in the latter horizon.

For any box $i \in [n]$, let $\alpha_i$ denote the time when box $i$ arrives for the first time in Poisson Rounding. 
The algorithm opens the boxes in the real time horizon in the ascending order of $\alpha_i$. 
Let $t_i$ denote the start time of box $i$ in the real time horizon. 
We have:
\[
    t_i ~= \sum_{j \,:\, \alpha_j < \alpha_i} c_j
    ~.
\]

Figure~\ref{fig:pseudo to real} illustrates how arrivals in the Poisson time horizon map to box openings in the real time horizon.

\begin{figure}[h]
    \centering
    \begin{tikzpicture}[scale = 0.9]
        \draw [thick,->] (2,0) -- (14.6,0);
        \draw [thick,-] (2,0.15) -- (2,0);
        \draw [thick,->] (4,0.8) -- (4,0.2);
        \draw [thick,-] (4,0.15) -- (4,0);
        \draw [thick,->] (7,0.8) -- (7,0.2);
        \draw [thick,-] (7,0.15) -- (7,0);
        \draw [thick,->] (12,0.8) -- (12,0.2);
        \draw [thick,-] (12,0.15) -- (12,0);
        \node[align=center] at (2,-0.5) {$0$};
        \node[align=center] at (4,-0.5) {$\alpha_3$};
        \node[align=center] at (4,1.4) {box $3$ \\ arrives};
        \node[align=center] at (7,-0.5) {$\alpha_5$};
        \node[align=center] at (7,1.4) {box $5$ \\ arrives};
        \node[align=center] at (12,-0.5) {$\alpha_1$};
        \node[align=center] at (12,1.4) {box $1$ \\ arrives};
        \draw [thick,->] (2,-3.5) -- (14.6,-3.5);
        \draw [thick,-] (2,-3.35) -- (2,-3.5);
        \draw [thick,-] (5,-3.35) -- (5,-3.5);
        \draw [thick,-] (9,-3.35) -- (9,-3.5);
        \draw [thick,-] (14,-3.35) -- (14,-3.5);
        \draw [thick,decorate,decoration={brace,amplitude=6pt,raise=0pt}] (2,-3.3) -- (5,-3.3);
        \draw [thick,decorate,decoration={brace,amplitude=6pt,raise=0pt}] (5,-3.3) -- (9,-3.3);
        \draw [thick,decorate,decoration={brace,amplitude=6pt,raise=0pt}] (9,-3.3) -- (14,-3.3);
        \node[align=center] at (2,-4) {$t_3 = 0$};
        \node[align=center] at (5,-4) {$t_5 = c_3$};
        \node[align=center] at (9,-4) {$t_1 = c_3 + c_5$};
        \node[align=center] at (14,-4) {$c_3 + c_5 + c_1$};
        \node[align=center] at (3.5,-2.4) {box $3$ is \\ being opened};
        \node[align=center] at (7,-2.4) {box $5$ is \\ being opened};
        \node[align=center] at (11.5,-2.4) {box $1$ is \\ being opened};
        \node at (16.7,-3.5) {real time horizon};
        \node at (17,0) {Poisson time horizon};
    \end{tikzpicture}
\caption{Illustration of events in the Poisson and real time horizons}
\label{fig:pseudo to real} 
\end{figure}
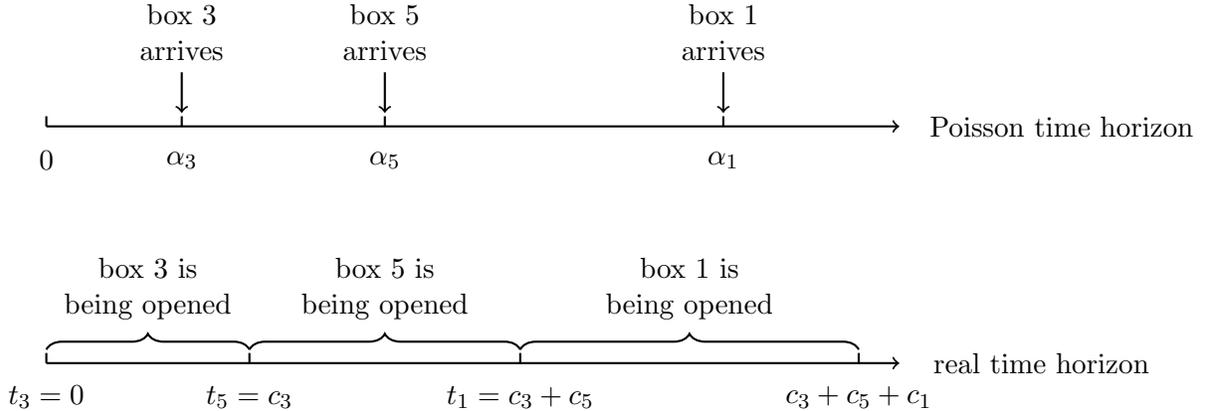

\noindent {\bf Some Useful Facts.~}
We first show that the real time horizon elapses at most as fast as the Poisson time horizon in expectation.
This is analogous to \Cref{lem:unit-cost-two-horizons} in the unit-cost case.

\begin{lemma}
	\label{lem:general-two-horizons}
    For any $\tau \ge 0$, we spend at most $\tau$ time in expectation in the real time horizon to open the boxes that arrive before $\tau$ in the Poisson time horizon, i.e.,
	$\E \left[\sum_{i \,:\, \alpha_i < \tau} c_i \right] 
    ~\le~ 
    \tau
    ~.$
\end{lemma}

\begin{proof}
By definition, the arrival rate of box $i$ at any Poisson time $\tau'$ is $\frac{1}{c_i} \bar{x}_i(\frac{\tau'}{2})$. As it takes $c_i$ to open box $i$ each time it arrives, the delay caused by box $i$ is at most:
    \[
        \int_0^{\tau} \frac{1}{c_i} \bar{x}_i\Big(\frac{\tau'}{2}\Big) \cdot c_i \dif{\tau'} ~=~ \int_0^{\tau} \bar{x}_i\Big(\frac{\tau'}{2}\Big) \dif{\tau'}
        ~.
    \]
        
    Summing over all boxes $i \in [n]$, the expected real time that the algorithm uses to open boxes arriving before time $\tau$ in the Poisson time horizon is at most:
    \[
   		\int_0^{\tau} \sum_{i \in [n]} \bar{x}_i \Big(\frac{\tau'}{2}\Big) \dif{\tau'} ~\underset{\eqref{eqn:unit-poisson-rate}}{\leq}~ \int_0^{\tau} 1 \dif{\tau'} ~=~ \tau
		~.
    \qedhere\]    
\end{proof}

Suppose that a stopping rule takes box $i$ at time $\tau$ in the Poisson time horizon. 
To upper bound the expected objective, we need to bound the amount of real time spent on opening boxes $j \ne i$ that arrive before time $\tau$ in the Poisson time horizon. 
For all stopping rules in the rest of the paper, the condition of taking box $i$ at time $\tau$  will be equivalent to all other boxes $j$ arriving later than some respective thresholds time $\theta_j$.
The lemma below provides a conditional upper bound on the expected time required to open these boxes, and is analogous to \Cref{lem:unit-cost-completion-time}.

\begin{lemma}
	\label{lem:general-completion-time}
	For any box $i \in [n]$, any time $\tau$, and any thresholds $\theta_j$ for $j \ne i$:
	\[
		\E \bigg[ \sum_{j \ne i \,:\, \alpha_j < \tau} c_j \mid\, \forall j \in [n], \alpha_j > \theta_j \,\bigg] 
        ~\le~ 
        \tau
        ~.
	\]
\end{lemma}

We cannot include $i$ in the above summation because the condition of taking $i$ at time $\tau$ implies $\alpha_i \le \tau$, which is in the opposite direction compared to the conditions for the other boxes.

\begin{proof}
	For each box $j \ne i$, conditioned on $\alpha_j \ge \theta_j$, its arrival rate at time $\tau'$, is $0$ if $\tau' \leq \theta_j$, or $\frac{1}{c_j}\bar{x}_j(\frac{\tau'}{2})$ if $\tau' > \theta_j$. 
	Importantly, the conditional arrival rates are at most the unconditional rates in both cases. 
	Hence, the conditional expectation is upper bounded by the unconditional expectation, which is at most $\tau$ by \Cref{lem:general-two-horizons}.
\end{proof}

Finally, our analysis needs to upper bound the probability that the algorithm does not take any boxes before some time threshold, which means that all boxes $i \in [n]$ arrive later than their corresponding thresholds $\theta_i$.
The following lemma is analogous to \Cref{lem:unit-cost-poisson-bound}, and follows from the definition of the Poisson process.

\begin{lemma}
	\label{lem:poisson-bound}
	For any non-negative thresholds $(\theta_i)_{i \in [n]}$ in the Poisson time horizon, we have:
	\[
		\Pr \big[\, \forall i \in [n], \alpha_i > \theta_i \,\big] 
        ~=~ 
        \exp \bigg( - \sum_{i \in [n]} \int_0^{\theta_i} \frac{1}{c_i} \bar{x}_i \Big( \frac{\tau}{2} \Big) \dif{\tau} \bigg)
		~.
	\]		
\end{lemma}

\subsection{Balanced Stopping}
\label{subsec:general-cost-balanced}

Having defined the Poisson Rounding, we now introduce the stopping algorithm for deciding when to stop and take the best opened box.
We refer to this rule as Balanced Stopping and define it within the Poisson time horizon, where we have nice independence properties. 
Of course, the bound on the objective value will apply to the real time horizon.

\medskip
\noindent
{\bf Balanced Stopping.~}
Fix any scenario $v$.
Define the \emph{balanced time} of each box $i \in [n]$ in the Poisson horizon to be:
\[
    \beta_i ~\defeq ~ c_i + v_i
    ~.
\]

In Balanced Stopping, a box $i$ can be taken at Poisson time $\tau$ if (i) it has arrived, i.e., $\tau \ge \alpha_i$, and (ii) $\tau$ is at least its balanced time, i.e., $\tau \ge \beta_i$.
The algorithm takes the first box satisfying these two conditions.
For a description consistent with the discussion in \Cref{sec:stopping-dicusssion}, define:
\begin{equation}
\label{eq:tau-stars}
  	\tau_i ~\defeq~ \max \big\{ \alpha_i, \beta_i \big\}
	\quad,\quad
	\tau^* = \min_{i \in [n]} \tau_i
	\quad,\quad
	i^* = \argmin_{i \in [n]} \tau_i
	~.  
\end{equation}

We call $\tau^*$ the stopping time, and $i^*$ the taken box.

\begin{algorithm}{Balanced Stopping with Poisson Rounding}
\begin{enumerate}%
    \item Let $\alpha_i$ be the earliest arrival time  of each box $i \in [n]$ according to Poisson Rounding.
    \item At time $\tau > 0$ in the Poisson horizon:
    \begin{enumerate}
        \item If $\tau = \alpha_i$ for some box $i$, open the box and reveal $v_i$, and thus, $\beta_i$ and $\tau_i$.
        \item If $\tau = \tau_{i^*}$ for some opened box $i^*$, %
        stop and take box $i^*$. %
    \end{enumerate}
\end{enumerate}
\end{algorithm}

We first prove a lemma that allows us to analyze Balanced Stopping with Poisson Rounding purely within the Poisson time horizon, where we enjoy the independence of various events. 

\begin{lemma}
	\label{lem:objective-poisson-horizon}
	Conditioned on any scenario $v$, any taken box $i^*$, and any stopping time $\tau^*$, the expected objective of Balanced Stopping is at most:
	\begin{equation}
		\label{eqn:objective-poisson-horizon}
        \tau^* + c_{i^*} + v_{i^*} ~=~ \tau^* + \beta_{i^*}
		~.			
	\end{equation}
\end{lemma}

\begin{proof}
	We prove a stronger claim that bounds the expectation further conditioned on the arrival time $\alpha_{i^*}$ of $i^*$.
	Observe that if $\tau^* > \beta_{i^*}$, then $\alpha_{i^*} = \tau^*$.
	Otherwise, if $\tau^*=\beta_{i^*}$, then $\alpha_{i^*} \leq \tau^*$.
	
	Under the stated conditions, any box $j\neq i$ with $\beta_j \le \tau^*$ must arrive after time $\tau^*$ (the arrival time for boxes with $\beta_j > \tau^*$ can be arbitrary).
    Applying \Cref{lem:general-completion-time} with $\theta_j = \tau^*$ if $\beta_j \le \tau^*$ and $\theta_j = 0$ if $\beta_j > \tau^*$, the expected real time spent on opening boxes other than $i^*$ is at most $\tau^*$.
	
    The lemma follows by summing $\tau^*$ with the opening cost $c_{i^*}$ and volume $v_{i^*}$ of box $i^*$.
\end{proof}

\medskip
\noindent {\bf Balanced Times and Good Arrivals.~}
Intuitively, \Cref{eqn:objective-poisson-horizon} is the reason why we call $\beta_i = c_i + v_i$ the balanced time --- if the algorithm takes a box $i$ at its balanced time, then the ``exploration cost'' it spends on opening other boxes equals the ``exploitation cost'' it spends on opening and taking box $i$. 
However, this intuition does not explain why Balanced Stopping is ``lossless'' and still yields a $4$-approximation.

We now explain how this stopping rule arises naturally through a thought experiment. 
The point is that the objective in our benchmark $\textsc{General LP}$ depends on scenario-dependent variables $Z_i(t \,|\, v)$, while Poisson Rounding relies on the scenario-independent variables $X_i(t)$ (necessarily so, as the algorithm does not know the scenario).
Note that $\dif X_i(t) \geq \dif Z_i(t \,|\, v)$ by Constraint~\eqref{eq:general-lp-3}.
Imagine an adversary that keeps $\dif Z_i(t \,|\, v)$'s (and hence the LP objective) unchanged, but artificially increases $\dif X_i(t)$'s. In the rounding, this would lead to the arrival of ``bad'' boxes; our stopping rule should ensure that it does not hurt the algorithm to take such boxes (by missing a ``good'' box with a lower cost that might arrive later).%
\footnote{In other words, the adversary can only help the online algorithm by giving such ``bad'' boxes.}
The following lemma formalizes this thought experiment.

\begin{lemma}
	\label{lem:good-and-bad}
	Consider any good Poisson process and any bad Poisson process of box arrivals.
	Let the arrival rates of box $i$ in the two processes be $\lambda_i^g(\tau)$ and $\lambda_i^b(\tau)$. 
	Suppose the total arrival rate of the good process at any time $\tau > 0$ is bounded by:
	\begin{equation}
		\label{eqn:good-arrival-bound}
		\sum_{i \in [n]} \lambda^g_i(\tau) \le \frac{2}{\tau}
		~.
	\end{equation}

	Then, our stopping rule guarantees that the expected value in \Cref{eqn:objective-poisson-horizon} with both good and bad arrivals, is never higher than in the process with only good arrivals.
\end{lemma}

\begin{proof}
	We prove the lemma for every scenario $v$ by a coupling argument.
	Let $\ALG^{g+b}$ be the upper bound of Balanced Stopping's objective in \Cref{eqn:objective-poisson-horizon} when there are both good and bad arrivals.
	Let $\ALG^g$ be the objective with only good arrivals. 
	We couple these two random variables by defining their values using the same realization of good arrivals.
	
	In this proof, we define the variables $\alpha_i$, $\beta_i$, $\tau_i$, $\tau^*$, and $i^*$ in \eqref{eq:tau-stars} with respect to the good Poisson process. 
	We let $\tau^b$ and $i^b$ denote the stopping time and taken box in the bad Poisson process, i.e., the counterparts for $\tau^*$ and $i^*$.
	
	It suffices to show that the expectation of $\ALG^g$, conditioned on any realization of $\tau^b$ and $i^b$, is greater than or equal to the corresponding conditional expectation of $\ALG^{g+b}$, i.e.:
	\[
		\E \big[ \ALG^g \mid \tau^b, i^b \big] ~\ge~ \E \big[ \ALG^{g+b} \mid \tau^b, i^b \big] 
		~.
	\]
	
	For any $\tau^* < \tau^b$, conditioned on $\tau^*$ and any realization of the remaining randomness, we have $\ALG^g = \ALG^{g+b}$ because both processes stops at time $\tau^*$ and take the same box $i^*$.
	
	It remains to consider the expectation conditional on $\tau^* > \tau^b$.\footnote{We omit the zero-measure event of $\tau^* = \tau^b$.}
	That is:
	\[
		\E \big[ \ALG^g \mid \tau^b, i^b, \tau^* > \tau^b\big] ~\ge~ \E \big[ \ALG^{g+b} \mid \tau^b, i^b, \tau^* > \tau^b \big] 
		~.
	\]
	
    By definition, the algorithm with both good and bad arrivals takes box $i^b$ at Poisson time $\tau^b$. 
	Hence, the right-hand-side equals $\tau^b + \beta_{i^b} \le 2 \tau^b$, as we have $\beta_{i^b} \le \tau^b$ by the definition of Balanced Stopping.
	
    It remains to show that the left-hand-side is at least $2 \tau^b$.
	For that, we further fix the realization of box arrivals (if any) before time $\tau^b$ in the good Poisson process.
	Let:
	\[
		S = \big\{ i \in [n] : \alpha_i < \tau^b \big\}
	\]
	denote the set of boxes that arrive before time $\tau^b$.
	Since $\tau_i \ge \tau^* > \tau^b$, where recall that $\tau_i = \max \{ \beta_i, \alpha_i \}$, it must be the case that $\tau_i = \beta_i > \tau^b$ for any box $i \in S$.
	Let $\beta^* = \min_{i \in S} \beta_i$.
	
	We can now characterize the Balanced Stopping algorithm with only the good Poisson process as follows.
	Either (i) some box $i \notin S$ arrives at time $\tau^b < \tau < \beta^*$ and Balanced Stopping takes box $i$ before $\beta^*$, or (ii) the algorithm stops at $\tau^* = \beta^*$ and takes the corresponding box with $\beta_{i^*} = \beta^*$.
	
    In the former case, we can bound the algorithm's objective with $\ALG^g \ge \tau$, generously assuming that box $i$ has zero volume and the algorithm takes it immediately.
    
    In the latter case, we have $\ALG^g = 2 \beta^*$ by \Cref{lem:objective-poisson-horizon}.
	Note that each box $i \notin S$ still arrives with the original arrival rate $\lambda_i^g(\tau)$ at any time $\tau > \tau^b$, conditioned on not arriving before time $\tau^b$, a convenient property of Poisson processes.
    The total arrival rate of these boxes not in $S$ is at most the total arrival rate of all boxes, which is no more than $\frac{2}{\tau}$ by the lemma's assumption, i.e., \Cref{eqn:good-arrival-bound}.
	See \Cref{fig:alg-good-bound} for an illustration.

    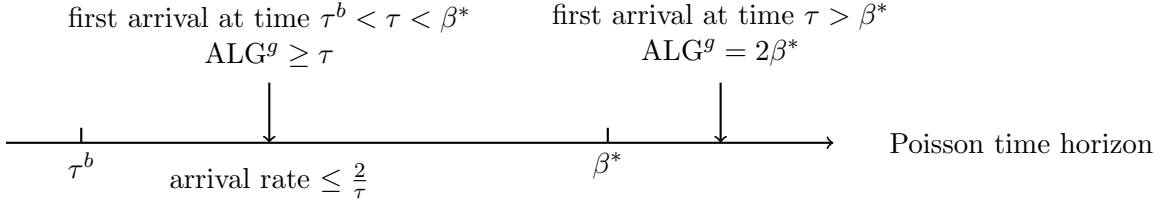
\begin{figure}[ht]
	\centering
	\begin{tikzpicture}
		\draw[thick,->] (0,0) -- (11,0);
		\draw[thick] (1,0.2) -- (1,0) node[pos=2.5] {$\tau^b$};
		\draw[thick] (8,0.2) -- (8,0) node[pos=2.5] {$\beta^*$};
		\node at (13.5,0) {Poisson time horizon};
		\draw[thick,<-] (3.5,0) -- (3.5,0.8) node[pos=1.8,align=center] {first arrival at time $\tau^b < \tau < \beta^*$\\ $\ALG^g\ge \tau$};
		\draw[thick,<-] (9.5,0) -- (9.5,0.8) node[pos=1.8,align=center] {first arrival at time $\tau >\beta^*$\\ $\ALG^g= 2\beta^*$};
		\node at (3.5,-0.5) {arrival rate $\le \frac{2}{\tau}$};
	\end{tikzpicture}	
	\caption{Illustration of the argument for bounding $\ALG^g$}
	\label{fig:alg-good-bound}
	\end{figure}

 	For any threshold $\theta > \tau^b$, the probability that there is no arrival between $\tau^b$ and $\theta$ is at most:
 	\[
 		\exp \Big( - \int_{\tau^b}^{\theta} \frac{2}{\tau} \dif{\tau} \Big) ~=~ \Big( \frac{\tau^b}{\theta} \Big)^2
 		~.
 	\]

	Putting the above into $\E \big[ \ALG^g \big] = \int_0^\infty \Pr \big[ \ALG \ge \theta \big] \dif{\theta}$, we get that:
	\[
		\E \big[ \ALG^g \big] 
		~\ge 
		\underbrace{\phantom{\bigg|} \tau^b}_{\text{for $0 < \theta < \tau^b$}}
		+~ 
		\int_{\tau^b}^{\beta^*} \big( \frac{\tau^b}{\theta} \big)^2 \dif{\theta} 
		~+~ 
		\beta^* \cdot \underbrace{\phantom{\bigg|} \Big( \frac{\tau^b}{\beta^*} \Big)^2}_{\text{for $\beta^* < \theta < 2\beta^*$}} 
		~=~ 
        2 \tau^b 
		~.
    \qedhere	
    \]
\end{proof}

\subsection{Proof of $4$-Approximation}
\label{subsec:general-approximation}

This subsection is devoted to proving the following main result of the paper:

\begin{theorem}\label{thm:general-cost-balanced}
    Balanced Stopping with Poisson Rounding is a $4$-approximation to the optimal partially adaptive algorithm for the Correlated Pandora's Problem.
\end{theorem}

We will prove a stronger claim that, for any fixed scenario $v$, the algorithm is a $4$-approximation to the corresponding LP objective.
Since we have fixed the scenario, we will write $Z_i(t)$ for $Z_i(t \,|\, v)$ for brevity.
Let:
\begin{equation}
    \label{eq:offline-cost}
        \CP_{\textsc{General}} (v) ~=~ \sum_{i \in [n]} \int_0^\infty (t+c_i+v_i) \dif{Z_i(t)}
\end{equation}
denote the objective of the \textsc{General LP} for scenario $v$.
We will prove that the expected objective of Balanced Stopping with Poisson Rounding in scenario $v$ is at most $4 \cdot \CP_{\textsc{General}}(v)$.

The value of $\CP_{\textsc{General}}(v)$ depends linearly only on variables $Z_i(t)$.
Therefore, we will first define a good arrival process using $Z_i(t)$ and then reduce the proof to an analysis of the good arrival process using \Cref{lem:good-and-bad}, and finally, use a series of inequalities to derive an upper bound of the algorithm's expected objective that is linear in $Z_i(t)$.
These ideas have been foreshadowed in \Cref{sec:unit-cost}.

\medskip
\noindent{\bf Reduction to the Good Arrival Process.~}
Recall that the arrival rate of any box $i \in [n]$ at any time $\tau$ in Poisson Rounding is:
\begin{align}
	\frac{1}{c_i} \bar{x}_i\Big(\frac{\tau}{2}\Big)
	&
	~=~ 
	\frac{2}{c_i \cdot \tau} \int_0^{\frac{\tau}{2}} \Big( X_i(t) - X_i\big( (t - c_i)_+ \big) \Big) \dif{t} \notag\\
	&
	~=~
	\frac{2}{c_i \cdot \tau} \int_0^{\frac{\tau}{2}} \min \Big\{ \frac{\tau}{2} - t, c_i \Big\} \dif{X_i(t)}	\label{eq:total-arrival-rate}
	~.
\end{align}

We define a {\em good} Poisson process of box arrivals with rates:
\begin{align*}
    \lambda^g_i(\tau) 
    &
    \,\defeq~ \frac{2}{c_i \cdot \max \{\tau, \beta_i\}} \int_0^{\frac{\tau}{2}} \Big( Z_i(t) - Z_i\big( (t - c_i)_+ \big) \Big) \dif{t} \\
    &
    ~=~\, \frac{2}{c_i \cdot \max \{\tau, \beta_i\}} \int_0^{\frac{\tau}{2}} \min \Big\{ \frac{\tau}{2} - t, c_i \Big\} \dif{Z_i(t)} 
    ~.
\end{align*}

Compared with \eqref{eq:total-arrival-rate}, notice that we have $\dif Z_i(\tau')$ here instead of $\dif X_i(\tau')$, and the denominator has the term $\max\{\tau,\beta_i\}$ instead of $\tau$. Both these changes reduce $\lambda_i^{g}(\tau)$ compared to $\bar{x}_i(\tau)$. 
The former allows us to compare the algorithm's performance with the LP benchmark, which is expressed in terms of $Z_i(t)$'s.
It further ensures that the good arrival rates satisfy the condition of \Cref{lem:good-and-bad}.

The reason for replacing $\tau$ by $\max(\tau,\beta_i)$ is more subtle (but crucial).
Comparing the definition of $\lambda_i^g$ and $\bar{x}_i$, some readers may wonder why we replace $\tau$ with $\max \{\tau, \beta_i\}$ in the definition of the former. 
The reason is that, with the original denominator $\tau$, an $\varepsilon$ mass of the LP solution could lead to an arbitrarily large probability of opening and taking a box in the good Poisson process.
This would make the comparison with the LP benchmark difficult.

Consider the following scenario. 
Box $1$ has $c_1 = \delta$, $\beta_1 = 1$, and $Z_1(\tau) = \varepsilon$ for all $\tau \ge 0$, i.e., with probability $\varepsilon$, the LP solution opens box $1$ at time $0$ and eventually takes it.
Box $2$ has $Z_2(\tau) = 0$ for $0 \le \tau < 1$ and $Z_2(\tau) = 1 - \varepsilon$ for $\tau \ge 2$.
It is a dummy for filling up the remaining $1-\varepsilon$ probability mass in $\sum_i Z_i(\infty) = 1$, but importantly, it never arrives before box $1$'s balanced point $\beta_1 = 1$.
Other boxes have infinite volumes and are irrelevant.
If box $1$ arrives before time $1$, it would be taken at its balanced point $\beta_1 = 1$. 
With the original denominator and small $\varepsilon$ and $\delta$, box $1$'s total arrival rate before time $1$, and approximately its arrival probability, equals:
\[
	\int_0^1 \bigg( \frac{2}{\delta \cdot \tau} \cdot \min \Big\{ \frac{\tau}{2}, \delta \Big\} \cdot \varepsilon \bigg) \dif\tau
	~=~ 
	2\Big(1 + \ln \frac{1}{2\delta}\Big) \varepsilon
	~\gg~ 
	\varepsilon
	~.
\]

In contrast, with the new denominator, the above expression becomes:
\[
	\int_0^1 \bigg( \frac{2}{\delta \cdot \max\{\tau, 1\}} \cdot \min \Big\{ \frac{\tau}{2}, \delta \Big\} \cdot \varepsilon \bigg) \dif\tau
	~=~
	\varepsilon \int_0^1 \bigg( \frac{2}{\delta} \cdot \min \Big\{\frac{\tau}{2}, \delta\Big\} \bigg) \dif{\tau}
	~<~ 2\varepsilon
	~.
\]

We remark that the modified denominator was not needed for MSSC where $c_i = 1$ and $v_i = 0$ or $\infty$, because the maximum operation is superfluous when $v_i = 0$, and the box is irrelevant when $v_i = \infty$. 

Accordingly, define a {\em bad} Poisson process of box arrivals with the difference of these rates:
    \[
	   \lambda_i^b(\tau) ~\defeq~\frac{1}{c_i} \bar{x}_i\Big( \frac{\tau}{2} \Big) - \lambda_i^g(\tau)
	   ~.
    \]

Notice that $\lambda_i^b(\tau)\geq 0$ because $\max \{\tau, \beta_i\} \ge \tau$ and $\dif X_i(\tau')  \ge \dif Z_i(\tau')$, where the latter follows from Constraint~\eqref{eq:general-lp-3} in \textsc{general LP}.

By \Cref{lem:good-and-bad}, the expected objective of Balanced Stopping (with both good and bad arrivals) is less than or equal to the counterpart with good arrivals only.
Hence, it remains to analyze the latter random process.
The arrival times in the remaining argument are defined w.r.t.\ the good Poisson process.

For each box $i \in [n]$, recall that $\alpha_i$ denotes the time when box $i$ arrives for the first time in the Poisson time horizon, and $\tau_i = \max \big\{ \alpha_i, \beta_i \big\}$ is the earliest time when box $i$ can be taken.
Further recall that the algorithm stops at time $\tau^* = \min_i \tau_i$ and takes box $i^* = \argmin_i \tau_i$.

In the following analysis of the algorithm, we will omit the superscript $g$ in $\lambda_i^g(\tau)$ for brevity, since this notation will be used repeatedly.

\medskip
\noindent {\bf Bounding the Objective.~}
By \Cref{eqn:objective-poisson-horizon}, it suffices to bound:
%
\[
	\E [ \tau^* + \beta_{i^*} ]
	~=~ 
	\E [\tau^*] + \E [\beta_{i^*}]
	~.
\]

The probability that the stopping time $\tau^*$ is greater than any threshold $\theta > 0$ is:
\begin{align*}
	\Pr \big[ \tau^* > \theta \big]
	&
	~=~ \Pr \big[\, \text{for any $i$ with $\beta_i \le \theta$, we have $\alpha_i > \theta$} \,\big] \\
	&
	~\le~ \exp \bigg( - \sum_{i \,:\, \beta_i \le \theta} \int_0^\theta \lambda_i(\tau) \dif{\tau} \bigg)
	~.
	\tag{by \Cref{lem:poisson-bound}}
\end{align*}

Thus, the expected stopping time is:
\begin{equation}
	\label{eqn:general-cost-opening-cost}
	\E[ \tau^*] 
	~=~ \int_0^\infty \, \Pr \big[ \tau^* \ge \theta \big] \dif{\theta}
	~\le~ \int_0^\infty \exp \bigg( - \sum_{i \,:\, \beta_i < \theta} \int_0^\theta \lambda_i(\tau) \dif{\tau} \bigg) \dif{\theta}
	~.
\end{equation}

Next, we analyze $\E[\beta_{i^*}]$.
For each box $i$, consider the probability it is taken, i.e., $\Pr \big[ i^* =i \big]$.
First, consider the case when box $i$ arrives at time $\alpha_i < \beta_i$ in the Poisson horizon, and the other boxes arrive too late so that the algorithm takes box $i$ at time $\beta_i$.
We upper bound the probability of this case by box $i$'s total arrival rate in this period:
\begin{equation}
	\label{eqn:general-cost-volume-early}
	\int_0^{\beta_i} \lambda_i(\tau) \dif{\tau}
	~.
\end{equation}

If box $i$ arrives at some time $\alpha_i > \beta_i$, a necessary condition for it to be taken is that no box $j$ with $\beta_j < \alpha_i$ arrives before time $\alpha_i$. 
The probability of taking box $i$ in this case equals:
\begin{equation}
	\label{eqn:general-cost-volume-late}
	\int_{\beta_i}^\infty \underbrace{\vphantom{\sum_{j: s_j < \theta}} \lambda_i(\theta)}_{\text{arrival rate of $i$}} \,\cdot\quad  \underbrace{\exp \bigg( - \sum_{j: \beta_j < \theta} \int_0^{\theta} \lambda_j(\tau) \dif{\tau} \bigg)}_{\Pr [\text{no other box taken before $\theta$}]} ~\dif{\theta}
	~.
\end{equation}

We bound $\E[\beta_{i^*}]$ by summing \Cref{eqn:general-cost-volume-early,eqn:general-cost-volume-late}, multiplying the result with $\beta_i$, and summing over all boxes $i \in [n]$.

Putting everything together, the algorithm's expected objective for scenario $v$ is at most:
%
\[
	\eqref{eqn:general-cost-opening-cost}+ \sum_{i \in [n]} \big( \eqref{eqn:general-cost-volume-early} + \eqref{eqn:general-cost-volume-late} \big) \cdot \beta_i
	~.
\]
%

Grouping the terms in $\eqref{eqn:general-cost-opening-cost}$ and $\sum_{i \in [n]} \eqref{eqn:general-cost-volume-late} \cdot \beta_i$, we can bound $ \E[\tau^*] + \E[\beta_{i^*}]$ as:
\begin{equation*}
	\int_0^\infty ~ \bigg( 1 + \sum_{i \,:\, \beta_i < \theta} \, \lambda_i(\theta) \cdot \beta_i \bigg) \cdot \exp \bigg( - \sum_{j \,:\, \beta_j < \theta} \int_0^{\theta} \lambda_j(\tau) \dif{\tau} \bigg) \dif{\theta}
	~+~
    \sum_{i \in [n]} \beta_i \int_0^{\beta_i} \lambda_i(\tau) \dif{\tau}
    ~.
\end{equation*}

We replace the index $j$ with $i$ above, as the two summations in the first term are independent.
In other words, we rewrite the above as:
%
\begin{equation}
	\label{eqn:general-total-cost}	
	\int_0^\infty ~ \bigg( 1 + \sum_{i \,:\, \beta_i < \theta} \, \lambda_i(\theta) \cdot \beta_i \bigg) \cdot \exp \bigg( - \sum_{i \,:\, \beta_i < \theta} \int_0^{\theta} \lambda_i(\tau) \dif{\tau} \bigg) \dif{\theta}
	~+~
	\sum_{i \in [n]} \beta_i \int_0^{\beta_i} \lambda_i(\tau) \dif{\tau}
    ~.
\end{equation}

\medskip
\noindent {\bf Linearizing the Bound.~}
Our next goal is to compare the bound in \Cref{eqn:general-total-cost} with the LP cost in \Cref{eq:offline-cost}. To do so, we first linearize the bound w.r.t.\ $Z_i(t)$, through several steps. 
The second part of \eqref{eqn:general-total-cost} is already linear in $Z_i(t)$ as $\lambda_i(t)$ is linear in $Z_i(t)$:
%
\begin{align}
    \sum_{i \in [n]} \beta_i \int_0^{\beta_i} \lambda_i(\tau) \dif{\tau} 
    &
    ~=~ \sum_{i \in [n]} \beta_i \int_0^{\beta_i} \frac{2}{c_i \cdot \max \{\tau, \beta_i\}} \int_0^{\frac{\tau}{2}} \min \Big\{\frac{\tau}{2} - t, c_i \Big\} \dif{Z_i(t)}  \dif{\tau} \notag \\
    \text{\footnotesize ($\tau \le \beta_i$)}
    &
    ~=~ 2 \sum_{i \in [n]} \int_0^{\beta_i} \int_0^{\frac{\tau}{2}} \frac{\min\{\frac{\tau}{2}-t, c_i\}}{c_i} \dif{Z_i(t)}  \dif{\tau} 
     \notag\\
    \text{\footnotesize (changing order of integrations)} 
    &
    ~=~ 2 \sum_{i \in [n]} \int_0^{\frac{\beta_i}{2}} \int_{2t}^{\beta_i} \frac{\min\{\frac{\tau}{2}-t, c_i\}}{c_i} \dif{\tau} \dif{Z_i(t)} 
	~.
    \label{eq:general-second-term}
\end{align}

Define function:
\begin{equation*}
    h(t, c, \beta) 
    ~\defeq~
    4 \int_{\min\{t, \beta\}}^{\beta} \frac{ \min\{\tau-t, c\}}{c} \dif{\tau} 
    ~.
\end{equation*}
%

Then, noting that $h(t,c,\beta)=0$ for $t \geq \beta$, we can rewrite \eqref{eq:general-second-term} as:
%
%

%
\begin{equation}
	\label{eqn:linear-part-2}	
    \sum_{i \in [n]} \int_0^\infty h\Big( t, c_i, \frac{\beta_i}{2} \Big) \dif{Z_i(t)}
    ~.
\end{equation}


We next consider the first part of \Cref{eqn:general-total-cost}.
As $1 + x \le e^x$, the integrand is at most:

\begin{equation}
    	\exp \bigg(\sum_{i \,:\, \beta_i < \theta} \, \lambda_i(\theta) \cdot \beta_i \,- \sum_{i \,:\, \beta_i < \theta} \int_0^{\theta} \lambda_i(\tau) \dif{\tau} \bigg)
	~.
\label{eq:first-integrand}
\end{equation}

Substituting $\lambda_i(\tau)$ with its definition, the exponent becomes:
\[
    \sum_{i \,:\, \beta_i < \theta} \int_0^{\frac{\theta}{2}} \frac{2\beta_i}{\theta} \cdot \frac{\min\{ \frac{\theta}{2} - t,c_i\}}{c_i} \dif{Z_i(t)} 
    ~-~
    \sum_{i \,:\, \beta_i < \theta} \int_0^{\theta} \frac{2}{\max\{\tau, \beta_i\}} \int_0^{\frac{\tau}{2}} \frac{\min\{\frac{\tau}{2} - t,c_i\}}{c_i} \dif{Z_i(t)} \dif{\tau}
    ~.
\]

Changing the order of integrations in the second part as in \Cref{eq:general-second-term}, and merging the two parts, we get that:
    \[
        \sum_{i \,:\, \beta_i<\theta} \int_0^{\frac{\theta}{2}} \bigg( \frac{2\beta_i}{\theta} \cdot \frac{\min\{\frac{\theta}{2} - t,c_i\}}{c_i} - \int_{2t}^\theta \frac{2 \cdot \min\{\frac{\tau}{2}-t,c_i\}}{c_i \cdot \max\{\tau, \beta_i\}} \dif{\tau} \bigg)\dif{Z_i(t)}
	   ~.
    \]
    
By a change of variables in the inner integral, this is:
    \begin{equation}
    \label{eq:general-first-integrand}
        \sum_{i \,:\, \beta_i<\theta} \int_0^{\frac{\theta}{2}} \bigg( \frac{2\beta_i}{\theta} \cdot \frac{\min\{\frac{\theta}{2} - t,c_i\}}{c_i} - \int_{t}^{\theta/2} \frac{2 \cdot \min\{\tau-t,c_i\}}{c_i \cdot \max\{\tau, \beta_i/2\}} \dif{\tau} \bigg)\dif{Z_i(t)}
	   ~.
    \end{equation}

Define the function:
\[
	g(t, c, \beta, \theta) 
    ~\defeq~
    \begin{cases}
        0 & \mbox{if $0 \leq \theta < \max\{2t,\beta\}$;} \\[2ex]
        \displaystyle \frac{2\beta}{\theta} \cdot \frac{\min\{\theta - t, c\}}{c} - \int_t^{\theta} \frac{2 \cdot \min\{\tau-t,c\}}{c \cdot \max\{\tau, \beta\}} \dif{\tau} & \mbox{otherwise.}
    \end{cases}
\]

Using this we can write \Cref{eq:general-first-integrand} as:
\[
	\sum_{i \in [n]} \int_0^\infty g \Big(t, c_i, \frac{\beta_i}{2}, \frac{\theta}{2} \Big) \dif{Z_i(t)}
	~.
\]
and the 
integrand \eqref{eq:first-integrand} becomes:
\[
	\exp \bigg(\sum_{i \in [n]} \int_0^\infty g \Big(t, c_i, \frac{\beta_i}{2}, \frac{\theta}{2} \Big) \dif{Z_i(t)} \bigg)
    ~.
\]

We now linearize this in $Z_i(t)$.
Recall the LP constraint $\sum_{i \in [n]} Z_i(\infty) = 1$, or equivalently, $\sum_{i \in [n]} \int_{0}^\infty \dif Z_i(t) = 1$.
By Jensen's inequality, the above is at most:
\[
    \sum_{i \in [n]} \int_0^\infty e^{g (t, c_i, \frac{\beta_i}{2}, \frac{\theta}{2} )} \dif{Z_i(\tau)} 
    ~.
\]

Therefore, the first part of \Cref{eqn:general-total-cost} is upper bounded by:
\begin{equation}
	\label{eqn:linear-part-1}
   	\int_0^\infty \sum_{i \in [n]} \int_0^\infty e^{g (t, c_i, \frac{\beta_i}{2}, \frac{\theta}{2} )}  \dif{Z_i(\tau)} \dif{\theta}
    ~.
\end{equation}

Combining \Cref{eqn:general-total-cost,eqn:linear-part-2,eqn:linear-part-1}, the algorithm's expected objective is at most:
\begin{align}
	&
        \sum_{i \in [n]} \int_0^\infty \bigg( \int_0^\infty e^{g (t, c_i, \frac{\beta_i}{2}, \frac{\theta}{2} )}  \dif{\theta}
    ~+~
    \int_0^\infty h\Big( t, c_i, \frac{\beta_i}{2} \Big) \bigg)\dif{Z_i(\tau)} \notag \\
	& \qquad
	=~
        \sum_{i \in [n]} \int_0^\infty \bigg( 2 \int_0^\infty e^{g (t, c_i, \frac{\beta_i}{2}, \theta )}  \dif{\theta}
    ~+~
    \int_0^\infty h\Big( t, c_i, \frac{\beta_i}{2} \Big) \bigg)\dif{Z_i(\tau)}   
    \label{eq:final-general-cost}
    ~.
\end{align}

\medskip
\noindent {\bf Comparing with the LP Benchmark.~} 
As $c_i+v_i=\beta_i$, the  LP cost \eqref{eq:offline-cost} for scenario $v$ is
\begin{align*}
    \CP_{\textsc{General}} (v) ~=~ \sum_{i \in [n]} \int_0^\infty (t+\beta_i) \dif{Z_i(t)}.
\end{align*} 

Compared this expression with \Cref{eq:final-general-cost}, to prove the stated approximation ratio of $4$, 
it suffices to show that the coefficient of $\dif Z_i(\tau)$ in \Cref{eq:final-general-cost} is at most $4(t+\beta_i)$, which we formulate as the next lemma. Here, it will be convenient to write $\beta_i$ as $2\beta$.

\begin{lemma}\label{lem:coefficient}
    For any $t > 0$, any $c > 0$, and any $\beta = \frac{\beta_i}{2} \ge \frac{c}{2}$, the function
    \[
        F(t,c,\beta) ~\defeq~ 4t + 8\beta - 2\int_0^\infty e^{g(t, c, \beta, \theta)} \dif{\theta} - h(t, c, \beta)
    \]
    is nonnegative.
\end{lemma}
We prove \Cref{lem:coefficient} analytically, but unfortunately, the proof involves a tedious case analysis, depending on the relation of $t$, $c$, $\beta$, and $\theta$, due to the maximum and minimum operations in the definitions of $h$ and $g$.
We defer it to \Cref{app:balanced-coefficient}.

\Cref{fig:plot of F} below presents a numerical verification of \Cref{lem:coefficient}.
Here we normalize time $t$ to be $1$ since the function is homogeneous, i.e., $F(\gamma t, \gamma c, \gamma \beta) = \gamma \cdot F(t, c, \beta)$. 
\Cref{fig:small-scale-verify} presents case when both $c$ and $\beta$ are upper bounded by $1$, showing that the function takes minimum value $F(1, 0, 0) = 0$ at $c = \beta = 0$.
\Cref{fig:large-scale-verify} demonstrates that the function is further from zero as $c$ and $\beta$ increase. Intuitively, this shows that the tight approximation ratio of $4$ is achieved when $c=v=0$, which corresponds to the special case of the MSSC problem.

\begin{figure}[ht]
    \centering
    \begin{subfigure}{.4\textwidth}
    \includegraphics[width=\textwidth]{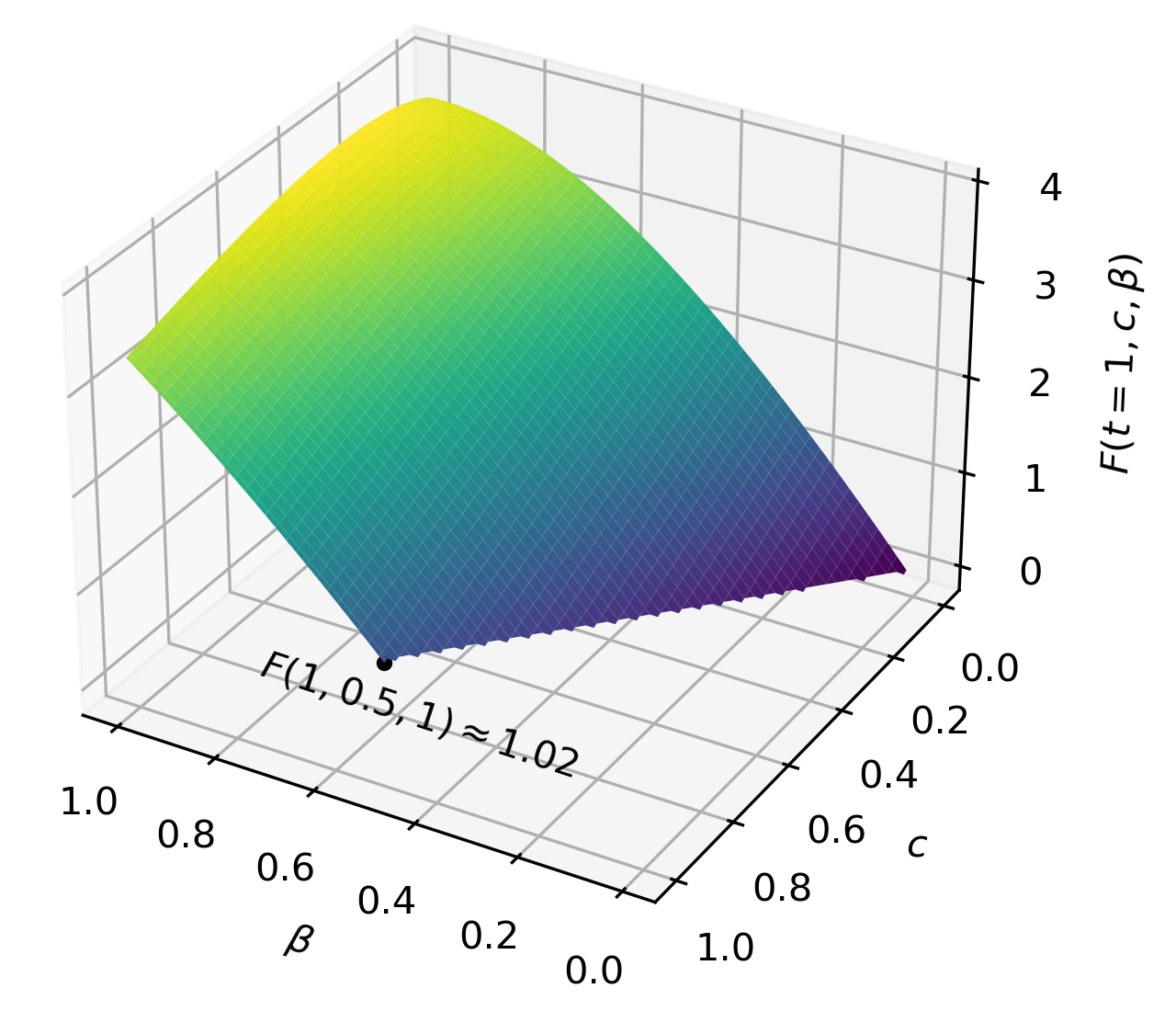}
    \caption{$0 \le c \le 1$ and $\frac{c}{2} \le \beta \le 1$}
    \label{fig:small-scale-verify}
    \end{subfigure}
	\begin{subfigure}{.4\textwidth}
    \includegraphics[width=\textwidth]{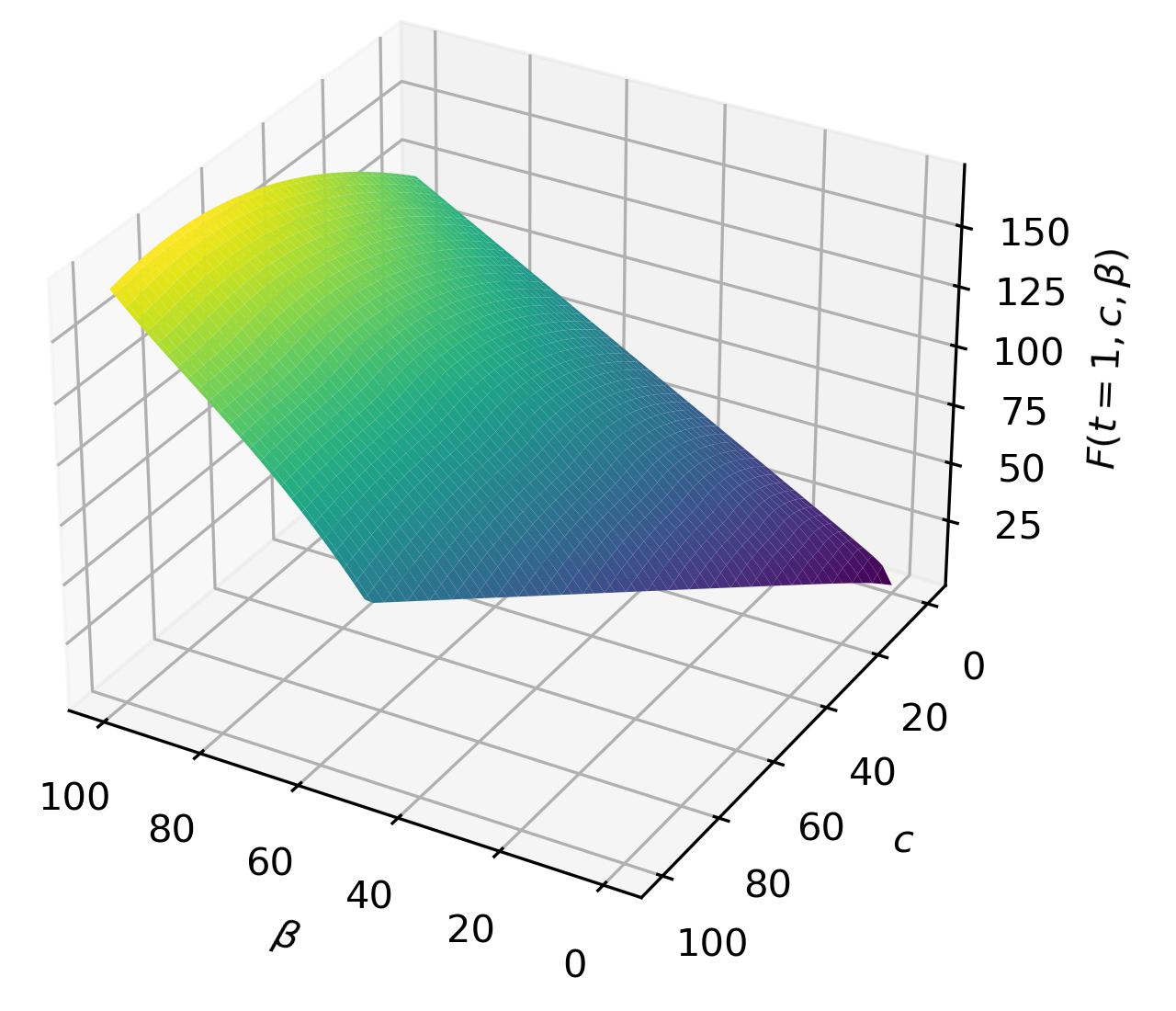}
    \caption{$0 \le c \le 100$ and $\frac{c}{2} \le \beta \le 100$}
    \label{fig:large-scale-verify}
    \end{subfigure}    
    \caption{Numerical plot of $F(1, \beta, c)$}
    \label{fig:plot of F}
\end{figure}
\section{Further Related Work}
\label{sec:related}

After \cite{Weitzman/1979/original} introduced Pandora's Problem, many variants have been studied in the literature. 
For the Correlated Pandora's Problem, we have already covered the negative results on approximating the optimal fully adaptive algorithm \citep{ChawlaGTTZ:FOCS:2020,DBLP:conf/approx/0001GMT23} and the existing results on approximating the optimal partially adaptive algorithm \citep{ChawlaGTTZ:FOCS:2020,GergatsouliT:NeurIPS:2023}.
Besides those, \cite{ChawlaGTTZ:FOCS:2020} considered the optimal non-adaptive algorithm as the benchmark, providing a $1.58$-approximation with a partially adaptive algorithm and a lower bound of $1.278$ that holds even for fully adaptive algorithms. 
\cite{GergatsouliT:NeurIPS:2023} formally pointed out that approximating the optimal partially adaptive algorithm with a factor strictly smaller than $4$ is NP-hard, even with the power of fully adaptive algorithms.

Other than this correlated model, most models of Pandora's Problem consider the maximization version, where inside each box is a value, and we aim to maximize the taken value minus the total box-opening cost.
We briefly discuss some of these models below, and refer readers to the recent survey by \citet{BeyhaghiC:SIGecom:2024}.

The optimality of Weitzman's Rule has been generalized in several directions.
\citet{DBLP:conf/soda/Singla18} proposed the Price of Information model as a generalization of Pandora's Problem in both cost-minimization and utility-maximization versions, where we may take a set of boxes subject to some feasibility constraints.
\citet{DBLP:conf/soda/Singla18} showed that Weitzman's Rule is optimal for matroid constraint and a $0.5$-approximation for knapsack constraint.
\cite{DBLP:conf/sigecom/KleinbergWW16} considered multiple stages of inspections.
In other words, each box needs to be inspected more than once before it is fully opened;
each inspection reveals further information about the final value and comes at a cost.
They proposed a generalized version of Weitzman's Rule and proved its optimality.

A widely studied variant is Pandora's Problem with Nonobligatory Inspection, where we can take a box without opening it and receive its expected value.
The model has been independently proposed several times in different literatures \citep*[e.g.,][]{GuhaMS:arXiv:2008,BeyhaghiK:EC:2019,DBLP:journals/jet/Doval18}.
The problem of finding the optimal strategy is NP-hard~\citep*{DBLP:conf/stoc/FuLL23} but has a PTAS \citep*{DBLP:conf/stoc/FuLL23,DBLP:conf/stoc/BeyhaghiC23}.
\cite{ZivL:corr:2024} extended this nonobligatory model from single-item selection into combinatorial selection problems, including matroid basis, matching, and facility location, by losing a factor of $4/3$.

Another variant is to restrict the way of opening the boxes. 
\citet{DBLP:conf/soda/Singla18} illustrated the power of Weitzman's Rule when the set of opened boxes must satisfy some feasibility constraints.
\cite{BoodaghiansFLL:Math.Oper:2023} proposed the order-constrained settings, where inspections must follow predefined sequences.
They proved the existence of optimal algorithms for tree-like order constraints and provided negative results for broader classes of order constraints.
A degenerate model studied by \cite{BergerEFF:EC:2023} introduces non-additive costs for box-opening, where submodular cost functions admit a non-adaptive optimal algorithm while the more general class of sub-additive (or even XOS) cost functions requires adaptive algorithms.

There are also models that consider other forms of exploration-exploitation tradeoffs.
In the Committed Pandora's Problem, the algorithm must decide whether to take a box or not right after it is opened. 
\citet*{FuLX:ICALP:2018} gave a PTAS for this model, and \citet{DBLP:conf/sigecom/Segev021} later designed an EPTAS.
In the Online Pandora's Problem, the boxes are presented online, and the algorithm chooses whether to open each box given its cost,
and then chooses irrevocably whether to keep the revealed value or pass on it.
\cite{EsfandiariHLM:AAAI:2019} provided a $0.5$-approximation to this problem and a $0.2$-approximation when we may take multiple boxes subject to a knapsack constraint.

\bibliographystyle{plainnat}
\bibliography{main}

\clearpage
\appendix

\section{Delayed Activation and Factor Revealing Linear Program}
\label{app:delayed-activation}

This section presents the Delayed Activation algorithm and its competitive analysis using a factor revealing LP. 
Since the resulting competitive ratio is subsumed by our main result, we will only consider the unit-cost special case in \Cref{sec:unit-cost} for ease of exposition.

Recall that Clairvoyant Stopping with (Discrete-Time) Poisson Rounding is a $4$-approximation (\Cref{thm:unit-cost-clairvoyant}), and this approximation ratio holds even if we need to pay $k$ times the actual volume for any $k \le 4$ (\Cref{remark:unit-cost-clairvoyant}).
Next, we design the $k$-\emph{Delayed Activation} algorithm to choose the same box as Clairvoyant Stopping, at the expense of increasing the objective by an additive factor proportional to $v_i$ for the taken box $i$.

\begin{algorithm}{$k$-Delayed Activation (with Poisson Rounding)}
\begin{itemize}[leftmargin=10pt]
    \item Sample arrival time $\alpha_i$ of boxes $i \in [n]$ using  (Discrete-Time) Poisson Rounding.
    \item Open boxes in ascending order of $\alpha_i$.
    \item Stop after time step $\min_i \alpha_i + \lfloor kv_i \rfloor$, and take box $i^* = \argmin_{i \in [n]} \big\{ \alpha_i + kv_i\big\}$.%
	    \footnote{This is implementable because box $i^*$ minimizes both $\alpha_{i^*} + \lfloor kv_{i^*} \rfloor$ and $\alpha_{i^*} + kv_{i^*}$, and has been opened at step $\alpha_{i^*}$.}
\end{itemize}
\end{algorithm}

The main result of this section is that $k$-Delayed Activation is $4.075$-competitive for the best (instance-dependent) choice of $k \le 4$.
To avoid the complication of finding the best $k$, we state the result for a randomized algorithm that is a distribution over $k$-Delayed Activation algorithms. 

\begin{theorem}\label{thm:delayed activation}
    Sampling $k \in [0, 4]$ with probability densities $\frac{e^k}{e^4 - 1}$, and then running $k$-Delayed Activation with Poisson Rounding is $\frac{4 e^4}{e^4 - 1} \approx 4.075$-competitive.
\end{theorem}

We start by analyzing the performance for a fixed $k \le 4$ and a fixed scenario.

\begin{lemma}
\label{lem:objective-DA}
	For any $k \le 4$, any scenario $v$, the objective of $k$-Delayed Activation is at most:
	\begin{equation}
		\label{eqn:objective-DA}
		\alpha_{i^*} + (k+1)v_{i^*} 
		~,
	\end{equation}
	where $i^* = \arg \min_{i \in [n]} \{ \alpha_i + kv_i\}$.
\end{lemma}

\begin{proof}
	By definition, the algorithm stops after time $\alpha_{i^*} + \lfloor kv_{i^*} \rfloor$ and takes the minimum volume from an opened box.
	The latter is at most $v_{i^*}$.
	Putting together these two parts, the objective is at most $\alpha_{i^*} + \lfloor kv_{i^*} \rfloor + v_{i^*} \le \alpha_{i^*} + (k+1)v_{i^*}$.
\end{proof}

To prepare for the final proof that considers different values of $k \le 4$ and different scenarios $v$, we introduce a few notations.
First, we define:
\[
	i_k^*(v) ~\defeq~ \argmin_{i \in [n]} \alpha_i + k \cdot v_i
	~.
\]

Further, let:
\[
	C(k) ~\defeq~ \E_{v \sim \Distribution} \big[\, \alpha_{i_k^*(v)} \,\big] ~, \quad V(k) ~\defeq~ \E_{v \sim \Distribution} \big[\, v_{i_k^*(v)} \,\big]
\]
denote the total box-opening cost and volume for taking box $i^*_k(v)$ right after it is opened at time $\alpha_{i^*_k(v)}$ in the real time horizon.
Finally, let $\ALG(k)$ be the expected objective of $k$-Delayed Activation. 
By \Cref{lem:objective-DA} and taking expectation over $v$, we get that:

\begin{lemma}
\label{lem:k-delayed-activation}
For any $k \le 4$, we have:
\[
	\ALG(k) ~\le~ C(k) + (k+1) \cdot V(k)
	~.
\]
\end{lemma}

By contrast, the conclusion of \Cref{remark:unit-cost-clairvoyant} can be written as:
\[
	C(k) + k \cdot V(k) ~\le~ 4 \cdot \OPT
	~.
\]

We cannot directly combine these two inequalities to get a $4$-approximation, because of the different coefficients of the $V(k)$ term.
Nonetheless, we will show that the expectation of $\ALG(k)$ for the stated distribution of $k$ is a $4$-approximation. 
The expected objective of the randomized algorithm can be written as follows, as a corollary of \Cref{lem:k-delayed-activation}.

\begin{corollary}
\label{cor:delayed-activation}	
We have:
\[
	\ALG ~=~ \int_0^4 \frac{e^k}{e^4-1} \cdot \Big( C(k) + (k+1) \cdot V(k) \Big) \dif{k}
	~.
\]
\end{corollary}

Finally, we characterize the relation of $C(k), V(k)$ for different values of $k \le 4$.

\begin{lemma}
	\label{lem:factor-revealing-different-k}
	For any $k, k' \le 4$, we have:
	\[
		C(k) + k \cdot V(k) ~\le~ C(k') + k \cdot V(k')
		~.
	\]	
\end{lemma}

\begin{proof}
	Consider any scenario $v$.
    Recall that $i_k^*(v)$ is the box that minimizes $\alpha_i + k \cdot v_i$.
    Hence:
    \[
        \alpha_{i_k^*(v)} + k \cdot v_{i_k^*(v)} ~\leq~ \alpha_{i_{k'}^*(v)} + k \cdot v_{i_{k'}^*(v)}
        ~.
    \]
    
    Taking expectation over $v \sim \Distribution$, and by the definition of $C(k)$ and $V(k)$, the lemma follows.
\end{proof}

\begin{proof}[Proof of \Cref{thm:delayed activation}]
	Consider any instance of the Correlated Pandora's Problem.
    Without loss of generality, we may assume that $\OPT = 1$ by scaling the box-opening costs and volumes by an appropriate multiplicative factor.
	Then, the expression for \Cref{remark:unit-cost-clairvoyant} with $k = 4$ simplifies as:
    \begin{equation}
    \label{eqn:lp constraint 1}
        C(4) + 4 \cdot V(4) \leq 4
        ~.
    \end{equation}
    
    By \Cref{cor:delayed-activation}, it suffices to show that:
    \[
    	\int_0^4 \frac{e^k}{e^4-1} \cdot \Big( C(k) + (k+1) \cdot V(k) \Big) \dif{k} ~\le~ \frac{4e^4}{e^4-1}
    	~.
    \]   
    
    Further, since $C(k)$ and $V(k)$ must satisfy the inequalities given in \Cref{eqn:lp constraint 1} and \Cref{lem:factor-revealing-different-k}, we only need to verify that the optimal objective value of the following \emph{factor revealing LP} is at most the stated competitive ratio.
	\begin{align*}
		\maximize \quad & \int_0^4 \frac{e^k}{e^4 - 1} \cdot \Big( C(k) + (k + 1) \cdot V(k) \Big) \dif{k} \\
		\subjectto \quad & C(4) + 4 \cdot V(4) \leq 4  \\[1ex]
			& C(k) + k \cdot V(k) \leq C(k') + k \cdot V(k') && \forall k, k' \in [0,4] \\[1ex]
		    & C(k) \geq 0, V(k) \geq 0  && \forall k \in [0,4] 
	\end{align*}

	The standard approach to bound the optimal objective value of a maximization LP is to construct a feasible assignment to its dual LP.
	However, this is an infinite-dimensional LP, and the optimal dual assignment involves Dirac delta functions. 
	To avoid such complications, we consider the discretized LP below, which converges to the factor revealing LP when $N$ tends to infinity.
    \begin{align*}
	\maximize \quad & \sum_{i=1}^N \frac{4}{N} \cdot \frac{e^\frac{4i}{N}}{e^4 - 1} \cdot \bigg( C_i + \Big(\frac{4i}{N} + 1\Big) \cdot V_i \bigg) \\
	\subjectto \quad & C_{N} + 4 \cdot V_N \leq 4  \\[.5ex]
		& C_i + \frac{4i}{N} \cdot V_i \leq C_{i+1} + \frac{4i}{N} \cdot V_{i+1} && \forall 1 \leq i \leq N-1 \\[1ex]
        & C_i \geq 0, V_i \geq 0  && \forall 1 \leq i \leq N 
    \end{align*} 
    
    Note that we have relaxed the second set of constraints by only keeping the inequalities for $k = \frac{4i}{N}$ and $k' = \frac{4(i+1)}{N}$. 
    This relaxation can only increase the optimal objective value, and thus, it is sufficient to bound the optimal objective value of the relaxed LP.
    The corresponding dual LP is:
    \begin{align*}
	\minimize \quad & \sum_{i=1}^N 4 P \\
	\subjectto \quad & Q_1 \geq \frac{4}{N} \cdot \frac{e^{\frac{4}{N}}}{e^4 - 1} \\
	& Q_i - Q_{i-1} \geq \frac{4}{N} \cdot \frac{e^{\frac{4i}{N}}}{e^4 - 1} && \forall 2 \leq i \leq N-1 \\
        & P - Q_{N-1} \geq \frac{4}{N} \cdot \frac{e^{4}}{e^4 - 1} \\
        & \frac{4}{N} \cdot Q_1 \geq \frac{4}{N} \cdot \frac{e^{\frac{4}{N}}}{e^4 - 1} \cdot \Big(\frac{4}{N} + 1\Big) \\
        & \frac{4i}{N} \cdot Q_i - \frac{4(i-1)}{N} \cdot Q_{i-1} \geq \frac{4}{N} \cdot \frac{e^{\frac{4i}{N}}}{e^4 - 1} \cdot \Big( \frac{4i}{N} + 1 \Big) && \forall 2 \leq i \leq N-1 \\
        & 4 P - \frac{4(N-1)}{N} \cdot Q_{N-1} \geq \frac{20}{N} \cdot \frac{e^{4}}{e^4 - 1} \\[1ex]
        & P \geq 0, Q_i \geq 0  && \forall 1 \leq i \leq N
    \end{align*}

	Consider a feasible dual assignment as follows:
    \[
    	P = \frac{1}{N(e^4-1)} \sum_{j=1}^N e^{\frac{4j}{N}} \Big( \frac{4j}{N} + 1 \Big)
    	~,\quad
    	Q_i = \frac{1}{i(e^4-1)} \sum_{j=1}^i e^{\frac{4j}{N}} \Big( \frac{4j}{N} + 1 \Big) 
        ~.
    \]
    
    The corresponding dual objective value is:
    \[
        \frac{4}{N(e^4-1)} \sum_{i=1}^N e^{\frac{4i}{N}} \Big(\frac{4i}{N} + 1 \Big)
        ~.
    \]
    
    By weak duality, this is an upper bound for the relaxed primal LP's optimal objective value.
    The limit when $N$ tends to infinity proves the stated competitive ratio.
\end{proof}

\section{Missing Proofs in \Cref{sec:general-cost}}
\label{app:general-cost}

\subsection{Proof of Lemma \ref{lem:general-cost-lp-solvability}}
\label{app:general-cost-lp-solvability}

	We present how to solve the program up to a $1 + O(\varepsilon)$ multiplicative factor with the understanding that scaling $\varepsilon$ by a constant factor only changes the polynomial bound by a (different) constant factor.
	Let $c_{\min} = \min_i c_i$ and $c_{\max} = \max_i c_i$.
	
	First, we round the opening costs and volumes of the boxes up to the closest multiple of $\varepsilon \cdot c_{\min}$, which we normalize to $1$ for expositional simplicity.
	The rounding transforms the costs and volumes into integers, while increasing the objective of any solution by at most a $1+2\varepsilon$ factor.
	Moreover, the time horizon is from $0$ to at most $\sum_i (c_i + \varepsilon \cdot c_{\min}) \le n \cdot \big( c_{\max} + \varepsilon \cdot c_{\min} \big)$, which is polynomial in $n$, $\frac{1}{\varepsilon}$, and $\frac{c_{\max}}{c_{\min}}$.
	
	As a result, we may consider $X_i(t)$ that are step functions with a fixed value within any time interval $[t-1, t)$ for positive integers $t$.
	In other words, we reduce the number of variables from uncountably many to polynomial.
	
	We can now solve the rounded program up to a $1 + \varepsilon$ factor using stochastic gradient descent with time and sample complexities polynomial in the number of variables and constraints, which are polynomial in $n$, $\frac{1}{\varepsilon}$, and $\frac{c_{\max}}{c_{\min}}$. 
	See \citet{bubeck:book} for a textbook reference for the convergence guarantees of stochastic gradient descent.
	\qed

\subsection{Proof of Lemma \ref{lem:coefficient}}
\label{app:balanced-coefficient}
We start by recalling the definitions of functions $g$, $h$, and $F$, and the lemma statement.
Function $g(t, c, \beta, \theta)$ is defined as:
\[
    g(t, c, \beta, \theta) =
    \begin{cases}
        0 & 0 \leq \theta < \mbox{if $\max\{t,\beta\}$;} \\[2ex]
        \displaystyle
        \frac{2 \beta}{\theta} \cdot \frac{\min\{\theta - t,c\}}{c} - \int_{t}^{\theta} \frac{2 \cdot \min\{t'-t,c\}}{c \cdot \max\{t',\beta\}} \dif{t'} & \mbox{otherwise.}
    \end{cases}
\]

For clarity, let us remove the minimum and maximum operations, and rewrite the function explicitly in six cases depending on the ordinal relation between $\beta$, $t$, $t+c$, and $\theta$.
We omit the basic calculus for deriving these explicit forms. 
\[
    g(t, c, \beta,\theta) 
    ~=~ 
    \begin{cases}
        0 &  0 \leq \theta < \max \{t, \beta\} \\[2ex]
        \displaystyle
        \frac{2 \beta (\theta - t)}{c\theta} -\frac{2(\theta - t)}{c}  + \frac{2t}{c} \ln \frac{\theta}{t} & \beta \leq t \le \theta \le t+c \\[2ex]
        \displaystyle
        \frac{2 \beta}{\theta} + \frac{2t}{c} \ln \frac{t+c}{t} - 2\ln \frac{\theta}{t+c} - 2 & \beta \leq t \le t+c < \theta \\[2ex]
        \displaystyle
        \frac{2 \beta (\theta - t)}{c\theta} - \frac{(\beta-t)^2}{c \beta} - \frac{2(\theta - \beta)}{c} + \frac{2t}{c} \ln \frac{\theta}{\beta} & t < \beta \leq \theta \le t+c \\[2ex]        
        \displaystyle
        \frac{2 \beta}{\theta} + \frac{\beta^2 - t^2}{c \beta} + \frac{2t}{c} \ln \frac{t+c}{\beta} - 2 \ln \frac{\theta}{t+c} - 2 & t < \beta \le t+c < \theta \\[2ex]
        \displaystyle
        \frac{2 \beta}{\theta} + \frac{2t+c}{\beta} - 2\ln \frac{\theta}{\beta} - 2 & t \le t+c < \beta \le \theta
    \end{cases}
\]

We will consider $\int_0^\infty e^{g(t, c, \beta, \theta)} \dif{\theta}$.
According to the above explicit expressions, this integral will be divided into two or three segments depending on the ordinal relation of $\beta$, $t$, and $t+c$: $\beta \le t \le t+c$ (\textbf{case 1}), $t < \beta \le t+c$ (\textbf{case 2}), and $t \le t+c < \beta$ (\textbf{case 3}).
The first expression is common in all three cases.
Together with the second and third expressions, they cover case 1. 
Together with the fourth and fifth expressions, they cover case 2.
Finally, together with the last expression, they cover case 3.

Next, recall that function $h(t,c,\beta)$ is:
\[
    h(t, c, \beta) ~=~ 
    4 \int_{\min\{t, \beta\}}^{\beta} \frac{\min\{t'-t, c\}}{c} \dif{t'} 
    ~.
\]

Once again, let us remove the minimum and maximum operations, and rewrite the function explicitly in the above three cases.
\[
    h(t,c,\beta)=
    \begin{cases}
        0 & \beta \leq t \le t + c \\[1ex]
        \displaystyle
        \frac{2(\beta - t)^2}{c} & t < \beta \le t+c \\[1.5ex]
        4\beta - 4t - 2c & t \le t+c < \beta
    \end{cases}
\]

Finally, recall that for any $t, c > 0$, and $\beta \geq \frac{c}{2}$, we want to show that:
\[
    F(t,c,\beta) ~=~ 4t + 8\beta - 2\int_0^\infty e^{g(t, c, \beta, \theta)} \dif{\theta} - h(t, c, \beta) ~\geq~ 0
    ~.
\]

We will now prove the inequality in the three cases separately.
Since the analysis is already quite long, we will omit the proofs of simple facts about univariate functions.

\subsubsection{Case 1: $\beta \leq t \le t+c$}
We have:
\begin{align*}
    F(t,c,\beta) 
    &
    ~=~ 2t + 8\beta - 2 \int_t^{t+c} \exp \Big(\frac{2 \beta (\theta - t)}{c\theta} -\frac{2(\theta - t)}{c}  + \frac{2t}{c} \ln \frac{\theta}{t}\Big) \dif{\theta} \\[1.5ex]
    &
    \qquad\qquad\qquad
    - 2 \int_{t+c}^\infty \exp \Big(\frac{2\beta}{\theta} + \frac{2t}{c} \ln \frac{t+c}{t} - 2 \ln \frac{\theta}{t+c} - 2 \Big) \dif{\theta} 
    ~,
\end{align*}
where we subtracted $2 \int_0^t e^{g(t,c,\beta,\theta)} \dif{\theta} = 2\int_0^t 1 \dif{\theta}$ from the $4t$ term in the expression of $F(t, c, \beta)$.

The last integral has a closed form because the part related to $\theta$ can be written as:
\[
    \int_{t+c}^\infty \exp \Big(\frac{2\beta}{\theta} - 2 \ln \frac{\theta}{t+c} \Big) \dif{\theta} = \frac{(t+c)^2}{2 \beta} \big( e^{\frac{2 \beta}{t + c}} - 1 \big)
    ~.
\]

Combining the with other parts, function $F$ further simplifies to:
\begin{equation}
	\label{eqn:case1-expression}
    2t + 8\beta - 2 \int_t^{t+c} \exp \Big(\frac{2 \beta (\theta - t)}{c\theta} -\frac{2(\theta - t)}{c}  + \frac{2t}{c} \ln \frac{\theta}{t} \Big) \dif{\theta} - \frac{(t+c)^2}{e^2 \beta} \big(e^{\frac{2\beta}{t+c}} -1 \big) \Big(1 + \frac{c}{t}\Big)^{\frac{2t}{c}}
    ~.
\end{equation}

The above expression is concave in $\beta$ because the last two (negative) terms are convex.
The integral is convex because the integrand is convex for any $\theta \ge t$ by the convexity of function $e^x$.
The last part is convex by the convexity of $\frac{e^x-1}{x}$.
Hence, it suffices to verify the non-negativity of the expression at the two endpoints $\beta = \frac{c}{2}$ and $\beta=t$.

\paragraph{Sub-case 1.1: $\beta = \frac{c}{2}$.}
\Cref{eqn:case1-expression} becomes:
\[
    2t + 4c - 2 \int_t^{t+c} \exp \Big(\frac{\theta - t}{\theta} -\frac{2(\theta - t)}{c}  + \frac{2t}{c} \ln \frac{\theta}{t} \Big) \dif{\theta} - \frac{2(t+c)^2}{e^2 c} \big(e^{\frac{c}{t+c}} -1 \big) \Big(1 + \frac{c}{t}\Big)^{\frac{2t}{c}}
    ~.
\]

Next, we change variables with $x = \frac{c}{t} \in [0,2]$, noting that $c \le 2\beta \le 2t$, and $y= \frac{\theta}{t} \in [1, 1+x]$ in the integral.
Further, we divide both sides by $2$. 
It remains to show that:
\begin{equation}
\label{eqn:case1.1}
    1 + 2x - \int_1^{1+x} e^{\frac{y-1}{y} - \frac{2(y-1)}{x}} y^{\frac{2}{x}} \dif{y} - \frac{(1+x)^2}{e^2x}(e^{\frac{x}{1+x}} - 1) (1+x)^{\frac{2}{x}} 
    ~\ge~ 0
    ~.
\end{equation}

It follows by combining the two inequalities below.

\begin{lemma}
	\label{lem:case1.1-a}
	For any $0 \le x \le 2$:
	\[
		\int_1^{1+x} e^{\frac{y-1}{y} - \frac{2(y-1)}{x}} y^{\frac{2}{x}} \dif{y} ~\le~ 1.28 x
		~.
	\]
\end{lemma}

\begin{lemma}
	\label{lem:case1.1-b}	
	For any $0 \le x \le 2$:	
	\[
		\frac{(1+x)^2}{e^2x}(e^{\frac{x}{1+x}} - 1) (1+x)^{\frac{2}{x}} ~\le~ 1 + 0.72x
		~.
	\]
\end{lemma}

\begin{proof}[Proof of \Cref{lem:case1.1-a}]
It suffices to verify that for any $1 \le y \le 1+x$: 
\[
    e^{\frac{y-1}{y} - \frac{2(y-1)}{x}} y^{\frac{2}{x}} \leq 1.28
    ~.
\]

The left-hand-side can be written as:
\[
	\big( y e^{-(y-1)} \big)^{\frac{2}{x}} \cdot e^{\frac{y-1}{y}} 
	~\le~ 
	\big( y e^{-(y-1)} \big) \cdot e^{\frac{y-1}{y}} 
	~=~ 
	e^{\ln y + 2-y-\frac{1}{y}}
	~.
\]
where the inequality follows by $x \le 2$ and  $y = 1 + (y-1) \le e^{y-1}$.
By the first-order condition of the exponent, the right-hand-side achieves its maximum value about $1.278 < 1.28$ at $y = \frac{\sqrt{5}+1}{2}$.
\end{proof}

\begin{proof}[Proof of \Cref{lem:case1.1-b}]
We start by relaxing $\frac{1+x}{x}(e^{\frac{x}{1+x}} - 1)$.
For $z = \frac{x}{1+x} \leq \frac{2}{3}$:
\[
    e^z - 1
    ~\leq~ 
    z + \frac{1}{2} z^2 + \frac{1}{5} z^3
    ~.
\]

We get that:
\[
    \frac{1+x}{x} \big( e^{\frac{x}{1+x}}-1 \big)
    ~\leq~ 
    1 + \frac{1}{2} \cdot \frac{x}{1+x} + \frac{1}{5} \cdot (\frac{x}{1+x})^2  
    ~=~ 
    1 + \frac{x}{2} \cdot \Big(1 - \frac{x(5x+3)}{5(1+x)^2} \Big) 
    ~\leq~ 
    1 + \frac{x}{2}
    ~.
\]

Next, we relax the $(1+x)^{\frac{2}{x}+2}$ term.
We state this simple fact formally below as it will be reused in the proof of another case.

\begin{fact}
	\label{fact:1}
	For any $0 \le x \le 2$:
	\[
	 	(1+x)^{\frac{2}{x} + 0.82} \leq e^2
	 	~.
	 \]
\end{fact}

This function first decreases and then increases in $0 \le x \le 2$.
We choose the constant $0.82$ such that the value at $x = 2$ is $3^{1.82} \approx 7.385$, only slightly smaller than $e^2 \approx 7.389$.

Putting these two inequalities to the left-hand-side of the lemma, it remains to show that:
\[
    (1+x)^{0.18} \big(2 + x \big) \leq 2 + 1.44x
    ~.
\]

The left-hand-side is convex in $x \ge 0$.
This becomes apparent after changing variables with $y = 1 + x$ and letting $a = 0.18$, since $y^a + y^{1+a}$ is convex in $y \ge 1$ for any $a \ge 0$.
Hence, it remains to verify the inequality at $x = 0$ (with equality) and $x = 2$ (left $\approx 4.564 < 4.88 =$ right).
\end{proof}

\paragraph{Sub-case 1.2: $\beta= t$.}
\Cref{eqn:case1-expression} becomes:
\[
    10t - 2 \int_t^{t+c} \exp \Big(- \frac{2(\theta-t)^2}{c \theta} + \frac{2t}{c} \ln \frac{\theta}{t} \Big) \dif{\theta} - \frac{(t+c)^2}{e^2 t} \big(e^{\frac{2t}{t+c}} -1 \big) \Big( 1 + \frac{c}{t} \Big)^{\frac{2t}{c}}
    ~.
\]

Changing variables with $x = \frac{c}{t} \in [0,2]$ and $y= \frac{\theta}{t} \in [1, 1+x]$, it remains to show that:
\begin{equation}
\label{eqn:case1.2}
    10 - 2\int_1^{1+x} e^{-\frac{2(y-1)^2}{xy}} y^{\frac{2}{x}} \dif{y} - \frac{1}{e^2} (e^{\frac{2}{1+x}} - 1)(1+x)^{\frac{2}{x}+2}
    ~\ge~ 0
    ~.
\end{equation}

First, observe that $(1+x)^{\frac{2}{x}+1}$ is increasing in $x$, and thus, is at most $9$.
Hence, the last term of \Cref{eqn:case1.2} is at most:
\[
    \frac{9}{e^2} (e^{\frac{2}{1+x}} - 1)(1+x)
    ~.
\]

Next, we consider the integral in \Cref{eqn:case1.2}.
Since $y^{\frac{2}{x}}$ is increasing in $y$ and $e^{-\frac{2(y-1)^2}{xy}}$ is decreasing in $y$, by Chebyshev's sum inequality, we have:
\[
    \int_1^{1+x} e^{-\frac{2(y-1)^2}{xy}} y^{\frac{2}{x}} \dif{y} ~\le~ \frac{1}{x} \int_1^{1+x} e^{-\frac{2(y-1)^2}{xy}} \dif{y} \int_1^{1+x} y^{\frac{2}{x}} \dif{y} 
    ~.
\]

Applying these two relaxations to \Cref{eqn:case1.2}, rearranging terms, and dividing them by $2$, the inequality reduces to:
\[
	\frac{1}{x} \int_1^{1+x} e^{-\frac{2(y-1)^2}{xy}} \dif{y} \int_1^{1+x} y^{\frac{2}{x}} \dif{y} 
	~+~
	\frac{9}{2e^2} (e^{\frac{2}{1+x}} - 1)(1+x) 
	~\le~ 
	5
	~.
\]

The second integral on the left-hand-side has a closed form:
\[
    \int_1^{1+x} y^{\frac{2}{x}} \dif{y} 
    ~=~ 
    \frac{x}{2+x}\Big((1+x)^{\frac{2}{x}+1} - 1\Big)
    ~\le~
    \frac{8x}{2+x}
    ~,
\]
where the inequality follows by $(1+x)^{\frac{2}{x}+1} \le 9$ for $x \le 2$ as argued above.

For $0 \le x \le 1$, we relax the first integral using $e^{-\frac{2(y-1)^2}{xy}} \leq 1$, the inequality becomes:
\[
	\frac{8x}{2+x} + \frac{9}{2e^2} (e^{\frac{2}{1+x}} - 1)(1+x) \le 5
	~.
\]

Multiplying both sides with $2+x$ and rearranging terms:
\[
	3x + \frac{9}{2e^2} (e^{\frac{2}{1+x}} - 1)(1+x)(2+x) \le 10
	~.
\]

The left-hand-side is convex in $x$, noting that $z(z+1)(e^{\frac{1}{z}} - 1)$ is convex in $z > 0$.
Further, the inequality holds at $x = 0$ (i.e., $\frac{9}{e^2} (e^2-1) \approx 7.78< 10$) and $x = 1$ (i.e., $3 + \frac{27}{e^2}(e-1) \approx 9.28< 10$).

For $1 \le x \le 2$, we bound the first integral by:
\[
    \int_1^{1+x} e^{-\frac{2(y-1)^2}{xy}} \dif{y} 
    ~\leq~ 
    \int_1^2 1 \dif{y} + \int_2^{1+x} e^{-\frac{y-1}{x}} \dif{y} 
    ~=~ 
    \frac{1}{x} + \frac{1}{e} - e^{-\frac{1}{x}}
    ~.
\]

The inequality becomes:
\[
	\frac{8x}{2+x}\big(\frac{1}{x} + \frac{1}{e} - e^{-\frac{1}{x}} \big) + \frac{9}{2e^2}(1+x)(e^{\frac{2}{1+x}} - 1) \le 5
    ~.
\]

The inequality holds at $x = 1$, as already verified in the previous case.
Further, the left-hand-side is decreasing in $x$.
The first term is decreasing because both $\frac{1}{2+x}$ and $x \big( \frac{1}{x} + \frac{1}{e} - e^{-\frac{1}{x}} \big)$ are decreasing, where the derivative of the latter is:
\[
    \frac{1}{e} - \Big(1-\frac{1}{x}\Big)e^{-\frac{1}{x}} > \frac{1}{e} - e^{-\frac{1}{x}} \geq 0
    ~.
\]

The second term is decreasing because $\frac{e^z - 1}{z}$ is increasing.

\subsubsection{Case 2: $t < \beta \le t+c$}
We have:
\begin{align*}
    F(t,c,\beta)
    &
    = 4t + 6\beta - 2 \int_{\beta}^{t+c} \exp \Big(\frac{2 \beta (\theta - t)}{c\theta} - \frac{(\beta-t)^2}{c \beta} - \frac{2(\theta - \beta)}{c} + \frac{2t}{c} \ln \frac{\theta}{\beta} \Big) \dif{\theta} \\[1.5ex]
    &
    \qquad -  2 \cdot \exp \Big( \frac{\beta^2-t^2}{c\beta} + \frac{2t}{c} \ln \frac{t+c}{\beta} - 2 \Big) \int_{t+c}^{\infty} \exp \left(\frac{2 \beta}{\theta}-2 \ln \frac{\theta}{t+c}\right) \dif{\theta} - \frac{2(\beta - t)^2}{c}
    ~,
\end{align*}
where we subtracted $2 \int_0^\beta e^{g(t,c,\beta,\theta)} \dif{\theta} = 2 \int_0^\beta 1 \dif{\theta}$ from the $8\beta$ term in the expression of $F(t, c, \beta)$.

The second integral has a closed form:
\[
    \int_{t+c}^{\infty} \exp \left(\frac{2 \beta}{\theta}-2 \ln \frac{\theta}{t+c}\right) \dif{\theta} = \frac{(t+c)^2}{2\beta}\big(e^{\frac{2\beta}{t+c}} - 1 \big)
    ~.
\]

Hence, function $F$ simplifies into:
\begin{align*}
    &
    4t + 6\beta - 2 \int_{\beta}^{t+c} \exp \Big(\frac{2 \beta (\theta - t)}{c\theta} - \frac{(\beta-t)^2}{c \beta} - \frac{2(\theta - \beta)}{c} + \frac{2t}{c} \ln \frac{\theta}{\beta} \Big) \dif{\theta} \\[1.5ex]
    &
    \qquad - \frac{(t+c)^2}{\beta} \big(e^{\frac{2\beta}{t+c}} - 1 \big) \exp \Big( \frac{\beta^2-t^2}{c\beta} + \frac{2t}{c} \ln \frac{t+c}{\beta} - 2 \Big) - \frac{2(\beta - t)^2}{c} 
    ~.
\end{align*}

Let $x = \frac{\beta}{t}$, $y = \frac{c}{t}$ and $z = \frac{\theta}{t}$, where:
\begin{equation}
	\label{eqn:case2-xy-range}
	1 \le x \le 1 + y ~,\quad y \le 2x 
	~.
\end{equation}

The equation becomes:
\begin{equation}
    \label{eqn:case2}
    \begin{aligned}
    &
    4 + 6x - 2 \int_x^{1+y} \exp \bigg(\frac{2x(z-1)}{yz} - \frac{(x-1)^2}{xy} - \frac{2(z-x)}{y} + \frac{2}{y} \ln \frac{z}{x} \bigg) \dif{z} \\[1.5ex]
    &
    \qquad - \frac{(1+y)^2}{x} \big(e^{\frac{2x}{1+y}} - 1 \big) \exp \Big( \frac{x^2-1}{xy} + \frac{2}{y} \ln \frac{1+y}{x} - 2 \Big) - \frac{2(x - 1)^2}{y} 
    ~.
    \end{aligned}
\end{equation}

Next, we show that the integral above is upper bounded by a much simpler form.

\begin{lemma}
\label{lem:case2-relaxation}
For any $x, y$ satisfying \Cref{eqn:case2-xy-range}, we have:
\begin{equation*}
    \int_x^{1+y} \exp \bigg(\frac{2x(z-1)}{yz} - \frac{(x-1)^2}{xy} - \frac{2(z-x)}{y} + \frac{2}{y} \ln \frac{z}{x} \bigg) \dif{z} 
    ~\leq~ 
    \frac{3}{2} + \frac{2y}{3}
    ~.
\end{equation*}
\end{lemma}

\begin{proof}
Rearranging the first three terms in the exponent of the integrand:
\begin{align*}
	\frac{2x(z-1)}{yz} - \frac{(x-1)^2}{xy} - \frac{2(z-x)}{y} 
	&
	~=~ 
	\bigg( \frac{2x(z-1)}{yz} - \frac{2(z-1)}{y} \bigg) - \bigg( \frac{(x-1)^2}{xy} + \frac{2(1-x)}{y} \bigg) \\
	&
	~=~ - \frac{2(z-1)(z-x)}{yz} + \frac{x^2-1}{xy} 
	~.
\end{align*}

We rewrite the integral as:
\[
    e^{\frac{x^2-1}{xy}} \int_x^{1+y} e^{-\frac{2(z-1)(z-x)}{zy}} \cdot \Big( \frac{z}{x} \Big)^{\frac{2}{y}} \dif{z}
    ~.
\]

Observe that ${\frac{2(z-1)(z - x)}{z y}}$ is increasing in $z$, with minimum value $0$ at $z = x$ and maximum value $2(1-\frac{x}{y+1}) \leq 2(1-\frac{x}{2x+1}) \leq \frac{4}{3}$ at $z = y+1$.
By $e^{-k} \leq 1 - \frac{k}{2}$ for $k \in [0,\frac{4}{3}]$, the above is at most:
\begin{align*}
    &
    e^{\frac{x^2-1}{xy}} \int_{x}^{1+y} \bigg(1 -\frac{(z-1)(z - x)}{z y}\bigg) \Big( \frac{z}{x} \Big)^{\frac{2}{y}} \dif{z}  \\[1.5ex]
    & \qquad 
    = \frac{1}{y} \bigg(\frac{e^{x-\frac{1}{x}}}{x^{2}} \bigg)^{\frac{1}{y}} \int_{x}^{1+y} \Big( (1+x+y)z^{\frac{2}{y}} - z^{\frac{2}{y}+1} - x z^{\frac{2}{y}-1} \Big) \dif{z} \\[1.5ex]
    & \qquad
    = \bigg(\frac{e^{x-\frac{1}{x}}}{x^{2}} \bigg)^{\frac{1}{y}} (1+x+y) \bigg[(1+y)^{\frac{2}{y}+1} - x^{\frac{2}{y}+1} \bigg] \frac{y}{2(1+y)(2+y)} %
    ~.
\end{align*}

Hence, it suffices to show that:
\begin{equation}
	\label{eqn:case2-b}
    \bigg(\frac{e^{x-\frac{1}{x}}}{x^{2}} \bigg)^{\frac{1}{y}} (1+x+y) \bigg[(1+y)^{\frac{2}{y}+1} - x^{\frac{2}{y}+1} \bigg] \frac{y}{2(1+y)(2+y)}   ~\le~ \frac{3}{2} + \frac{2y}{3}
    ~.
\end{equation}

We first prove the following lemma, which allows us to eliminate $x$ by replacing it with $\max \{1, \frac{y}{2} \}$ according to \Cref{eqn:case2-xy-range}.

\begin{lemma}
	\label{lem:case2-b-monotone}
	The left-hand-side of \Cref{eqn:case2-b} is non-increasing in $x \ge 1$.
\end{lemma}

\begin{proof}
The part involving $x$ is:
\[
    \bigg(\frac{e^{x-\frac{1}{x}}}{x^{2}} \bigg)^{\frac{1}{y}} (1+x+y) \bigg[(1+y)^{\frac{2}{y}+1} - x^{\frac{2}{y}+1} \bigg]
    ~.
\]

Consider its derivative:
\begin{align*}
    e^{\frac{1}{y}(x-\frac{1}{x})} \bigg[ 
    \underbrace{\vphantom{\Bigg|} \bigg( 1 + \frac{1}{y} \Big(1 - \frac{1}{x}\Big)^2 \bigg) (1+y)^{\frac{2}{y}+1} x^{-\frac{2}{y}}
    - 
    \bigg( \Big(1+\frac{1}{y}\Big) x - \frac{2}{y} + \frac{1}{y} \frac{1}{x} \bigg)
    - 
    (1+x+y) \Big( 1+\frac{2}{y} \Big)}_{(\star)}
    \bigg] 
    ~.
\end{align*}

The first term of $(\star)$ is decreasing in $x \ge 1$, while the (negative) second and third terms are increasing in $x \ge 1$. 
Hence, the three terms combined are decreasing in $x \ge 1$. 
It remains to verify that their combination is non-positive at $x = 1$ (if $y \le 2$) and $x = \frac{y}{2}$ (if $y > 2$).

In the former case, we have:
\[
    (\star) ~=~ (1+y)^{\frac{2}{y}+1} - 1 - \frac{(y+2)^2}{y} ~\le~ 3^2 - 1 - \frac{4^2}{2} ~=~ 0
    ~,
\]
where the inequality follows because $(1+y)^{\frac{2}{y}+1}$ is increasing and $\frac{(y+2)^2}{y}$ is decreasing for $y \leq 2$.

In the latter case, we have:
\[
	\frac{1}{x} \cdot (\star) ~=~ \bigg(1 + \frac{1}{2x} \Big(1-\frac{1}{x}\Big)^2 \bigg)  \bigg(\Big(2 + \frac{1}{x}\Big)^{\frac{1}{x}+1} - 1 \bigg) - \Big(3+\frac{1}{x}\Big) \Big(1+\frac{1}{x}\Big)
	~.
\]

Changing variable with $z = \frac{1}{x} \le 1$, the above equals:
\[
	\Big( 1 + \frac{1}{2} z (1-z)^2 \Big) \Big( (2+z)^{z+1} - 1 \Big) - (3+z)(1+z)
	~.
\]

Note that $(2+z)^{z+1} \le 2 + 7z$ with equality at $z = 0$ and $z = 1$. 
Further, relax the $z$ in front of $(1-z)^2$ to $1$ to simplify calculation.
The non-positivity of $(\star)$ reduces to:
\[
	\Big( 1 + \frac{1}{2} (1-z)^2 \Big) ( 1 + 7z ) - (3+z)(1+z) = \frac{1}{2} (1-z)(-3+8z-7z^2) \le 0
	~,
\]
where the last inequality follows by $3 + 7z^2 \ge 2\sqrt{21} z \ge 8z$.
\end{proof}

By \Cref{lem:case2-b-monotone}, it suffices to verify \Cref{eqn:case2-b} at $x = 1$ (if $y \le 2$) and $x = \frac{y}{2}$ (if $y > 2$).
In the former case, \Cref{eqn:case2-b} at $x = 1$ is:
\[
    \Big((1+y)^{\frac{2}{y} + 1} - 1\Big) \frac{y}{2(1+y)}  ~\le~ \frac{3}{2} + \frac{2y}{3}
    ~.
\]

Rearranging the terms, this is equivalent to:
\[
    (1+y)^{\frac{2}{y}+1} \le \frac{4y}{3} + \frac{3}{y} + \frac{16}{3} 
    ~.
\]

This holds because the left-hand-side is increasing in $y$ with maximum value $9$ at $y = 2$, while the right-hand-side achieves minimum value $4 + \frac{16}{3} > 9$ at $y = \frac{2}{3}$ by AM-GM.
In the latter case, \Cref{eqn:case2-b} at $x = \frac{y}{2}$ (written in $x$ for nicer coefficients):
\[
    e^{\frac{1}{2}-\frac{1}{2x^2}} \bigg(\Big(2 + \frac{1}{x}\Big)^{\frac{1}{x}+1} - 1 \bigg) \frac{x^2(1+3x)}{2(1+2x)(1+x)} 
    ~\le~
    \frac{3}{2} + \frac{4}{3} x 
    ~.
\]

Dividing both sides by $x$ and changing variables with $z = \frac{1}{x} \le 1$, it becomes:
\[
	e^{\frac{1}{2}(1-z^2)} \Big( (2+z)^{z+1} - 1 \Big) \frac{3+z}{2(2+z)(1+z)} ~\le~ \frac{3}{2} z + \frac{4}{3}
\]

By $e^{-\frac{1}{2}z^2} \le 1 - \frac{1}{2} z^2 + \frac{1}{8} z^4$, it suffices to show consider
\[
	(2+z)^{z+1} - 1 ~\le~ \frac{(9z+8)(z+1)(z+2)}{3\sqrt{e}(z+3)(1-\frac{1}{2}z^2+\frac{1}{8} z^4)}
\]

We relax $3\sqrt{e} \approx 4.95 < 5$ for nicer constants and move the $1$ to the right-hand-side: 
\[
	(2+z)^{z+1} 
	~\le~ 
	\frac{5z^5 + 15 z^4 + 52 z^3 + 220 z^2 + 376 z + 248}{5(z+3)(z^4 - 4z^2 + 8)}
	~.
\]

We relax the numerator by subtracting $\frac{4}{5} (2z-1)^2$ from it, so that the result is divisible by $z+2)$.
The inequality becomes:
\[
	(2+z)^z 
	~\le~ 
	\frac{25 z^4 + 25 z^3 + 210 z^2 + 648 z + 616}{25(z+3)(z^4 - 4z^2 + 8)}
	~.
\]

We relax the left-hand-side by:
\begin{align*}
	(2+z)^z 
	&
	~\le~ 1 + (\ln2) z  + \frac{1+(\ln2)^2}{2} z^2 + \frac{3-2\ln2-(\ln2)^2}{2} z^3 
	~.
\end{align*}
with equality at $z = 0$ and $z = 1$.
It becomes polynomial inequality:
\begin{align*}
	& 
	\frac{25}{2} \Big( (\ln2)^2 + 2\ln2 - 3 \Big) \cdot x^8 
	~+~ 
	\frac{25}{2} \Big( 2 (\ln2)^2 + 6 \ln 2 -10 \Big) \cdot x^7  \\
	& \qquad \vphantom{\bigg|}
	~-~
	\frac{25}{2} \Big( 7 (\ln2)^2 + 10 \ln 2 - 9 \Big) \cdot x^6 
	~-~
	25 \Big( 4 (\ln2)^2 + 15 \ln 2 - 19 \Big) \cdot x^5 \\
	& \qquad \vphantom{\bigg|}
	~+~
	50 \Big( 5 (\ln2)^2 + 6 \ln2 - 4 \Big) \cdot x^4
	~+~
	25 \Big( 8 (\ln2)^2 + 36 \ln2 - 35 \Big) \cdot x^3 \\
	& \qquad \vphantom{\bigg|}
	~-~
	10 \Big( 30 (\ln2)^2 + 20 \ln 2 - 21\Big) \cdot x^2
	~-~
	8 \Big( 75 \ln 2 - 56 \Big) \cdot x 
	~+~
	16
	~\ge~ 0
	~.
\end{align*}

Rounding every coefficient down to the closest multiple of $0.1$ and multiplying the resulting expression by $10$ for a polynomial with integer coefficients, it suffices to show that:
\begin{equation}
	\label{eqn:long-proof-start}
	- 142 x^8 - 611 x^7 - 162 x^6 + 1670 x^5 + 1280 x^4 - 1551 x^3 - 728 x^2 + 321 x + 160 ~\ge~ 0
	~.
\end{equation}

Next, we will repeatedly remove the highest-degree term with negative coefficients by adding a shifted version of Chebyshev's polynomial, and round the resuling polynomial's coefficients down to the closest integer.
\begin{align*}
& 
\mathrm{Equation}~\eqref{eqn:long-proof-start} \\[2ex]
& \qquad
\Uparrow 
\tag{add $\frac{142}{32768} \cdot \big( T_8 \big( 2z-1 \big) - 1 \big)$, and round} \\[2ex]
& 
-1179 z^7 + 761 z^6 + 889 z^5 + 1646 z^4 - 1645 z^3 - 717 z^2 + 320 z + 160 \ge 0 \\[2ex]
& \qquad
\Uparrow 
\tag{add $\frac{1179}{8096} \big( T_7(2z-1) - 1\big)$, and round} \\[2ex]
& 
-3366 z^6 + 6562 z^5 - 2223 z^4 - 291 z^3 - 943 z^2 + 334 z + 159 ~\ge~ 0 \\[2ex]
& \qquad
\Uparrow 
\tag{add $\frac{1683}{1024} \big( T_6(2z-1) - 1 \big)$, and round} \\[2ex]
& 
-3536 z^5 + 9137 z^4 - 6182 z^3 + 437 z^2 + 215 z + 159 ~\ge~ 0 \\[2ex]
& \qquad
\Uparrow 
\tag{add $\frac{221}{32} \big( T_5(2z-1) - 1 \big)$, and round} \\[2ex]
& \quad~
297 z^4 + 1553 z^3 - 2326 z^2 + 560 z + 145 ~\ge~ 0
~.
\end{align*}

At this point, we can consider the derivative of the left-hand-side, which is a degree-$3$ polynomial with a unique root $z \approx 0.702$ between $0$ and $1$.
At this root, the left-hand-side achieves its minimum value $\approx 1.24 > 0$.
\end{proof}

Applying \Cref{lem:case2-relaxation} to the left-hand-side of \Cref{eqn:case2}, it suffices to prove that:
\begin{equation}
\label{eqn:case2-reduced}
    1 + 6x - \frac{4y}{3} - \frac{(1+y)^2}{x} \big(e^{\frac{2x}{1+y}} - 1 \big) \exp \Big( \underbrace{\vphantom{\bigg|} \frac{x^2-1}{xy} + \frac{2}{y} \ln \frac{1+y}{x} - 2}_{p(x)} \Big) - \frac{2(x - 1)^2}{y} 
    ~\ge~
    0
    ~.
\end{equation}

This is concave in $x$.
It is easy to see that $x$ and $-(x-1)^2$ are concave. 
Further, $\frac{1}{x} \big(e^{\frac{2x}{1+y}} - 1 \big)$ is convex and increasing because $\frac{1}{z}(e^z - 1)$ is. 
Finally, we verify that the exponent $p(x)$ is increasing and convex, so that its composition with the exponential function is also increasing and  convex.
The first- and second-order derivatives of $p(x)$ are:
\begin{align*}
    p'(x) 
    = \frac{(x-1)^2}{y x^2} \geq 0
    ~, \quad
    p''(x)
    = \frac{2(x-1)}{y x^3} \geq 0
    ~.
\end{align*}

In sum, \Cref{eqn:case2-reduced} is concave in $x$ and we only need to check the inequality at the endpoints $x = \max\{1, \frac{y}{2}\}$ and $x = 1 + y$ of its feasible range given in \Cref{eqn:case2-xy-range}.

\paragraph{Proof of \Cref{eqn:case2-reduced} at $x = 1+y$.}
\Cref{eqn:case2-reduced} reduces to:
\[
    7 + \frac{8y}{3} - e^{-\frac{y}{y+1}} (y+1) (e^2-1)
    ~\ge~
    0
    ~.
\]

Relaxing $e^2-1$ into $7$ and re-arranging the terms, it suffices to prove that:
\[
    e^{-\frac{y}{y+1}} \le 1 - \frac{13}{21} \frac{y}{y+1}
    ~.
\]

This follows by $e^{-z} \le 1 - (1-\frac{1}{e}) z$ for $0 \le z \le 1$, noting that $1-\frac{1}{e} \approx 0.632 > 0.619 \approx \frac{13}{21}$.

\paragraph{Proof of \Cref{eqn:case2-reduced} at $x = 1$ ($y \le 2$).}
\Cref{eqn:case2-reduced} reduces to:
\[
    7 - \frac{4y}{3} - \frac{(y+1)^{2+\frac{2}{y}}}{e^2} \cdot (e^{\frac{2}{y+1}}-1) 
    ~\ge~
    0
    ~.
\]

We have $(1+y)^{0.82+\frac{2}{y}} \leq e^2$ by Fact~\ref{fact:1}.
Applying this to the above inequality, rearranging terms, and changing variables with $z = \frac{1}{1+y} \in [\frac{1}{3}, 1]$, the inequality reduces to:
\[
	e^{2z} 
	~\le~ 
	1 - \frac{4}{3} z^{0.18} + \frac{25}{3} z^{1.18}
    ~.
\]

By the convexity of $e^{2z}$:
\[
	e^{2z} < - 0.7 + 8.2z
	~,
\]
which can be verified at $x = \frac{1}{3}$ and $x = 1$.
Putting it into the inequality and rearranging terms again, the problem further reduces to proving:
\[
	z^{0.18} \ge \frac{246z -  51}{250z - 40}
	~.
\]

Further linearize the left-hand-side by the concavity of $z^{0.18}$, we have:
\[
	z^{0.18} \ge \frac{7 + 3z}{10}
	~,
\]
which can be verified at $z = 1$ and $z = \frac{1}{3}$.
The inequality then becomes a quadratic one:
\[
	75 z^2 - 83 z + 23 \ge 0
	~,
\]
which holds by AM-GM since $2 \sqrt{75 \times 23} > 83$.

\paragraph{Proof of \Cref{eqn:case2-reduced} at $x = \frac{y}{2}$ ($y > 2$).}
\Cref{eqn:case2-reduced} reduces to (written in $x > 1$):
\[
    \frac{7x}{3} + 3 - \frac{1}{x} - \frac{(1+2x)^2}{x} \cdot \big( e^{\frac{2x}{2x+1}}-1 \big) \cdot e^{-\frac{3}{2}-\frac{1}{2x^2} + \frac{1}{x} \ln(2 + \frac{1}{x})} 
    ~\ge~ 
    0
    ~.
\]

Dividing both sides by $x$ and changing variables with $z = \frac{1}{x} < 1$:
\[
	\frac{7}{3} + 3z- z^2 - (2+z)^2 \cdot \big( e^{\frac{2}{2+z}} - 1 \big) \cdot e^{-\frac{3}{2} - \frac{1}{2} z^2 + z \ln (2+z)}
	~\ge~
	0
	~.
\]

The exponent of the last term can be relaxed to $-3/2 + z \ln 2$ because:
\[
	z \ln (2 + z) 
	~=~ 
	z \ln 2 + z \ln \Big( 1 + z/2 \Big) 
	~\le~ 
	z \ln 2 + z^2/2
	~.
\]

It remains to show:
\[
	7/3 + 3z- z^2 
	~\ge~ 
	(2+z)^2 \cdot \big( e^{\frac{2}{2+z}} - 1 \big) \cdot 2^z \cdot e^{-3/2}
	~.
\]

The left-hand-side is concave and the right-hand-side is convex.
The convexity of $2^z$ is obvious, while the convexity of $(2+z)^2 \cdot \big( e^{\frac{2}{2+z}} - 1 \big)$ follows by, e.g., taking the Taylor expansion of function $e^x$ at $x = \frac{2}{2+z}$.
Finally, we can verify that the inequality holds at $z = 0$ (i.e., $7/3 > 4 (e-1) e^{3/2}$) and $z = 1$ (i.e., $13/3 > 18 (e^{\frac{2}{3}} - 1) e^{-3/2}$).

\subsubsection{Case 3: $t \leq t+c \leq \beta$}

We have:
\begin{align*}
    F(t,c,\beta)
    &
    ~=~ 4t + 8\beta - 2 \int_0^{\beta} 1 \dif{\theta} + 2 \int_{\beta}^{\infty} \exp \Big(\frac{2 \beta}{\theta} + \frac{2t+c}{\beta} - 2\ln \frac{\theta}{\beta} - 2\Big) \dif{\theta} - (4\beta - 4t - 2c) \\
	&
	~=~ 8t + 2c + 2\beta - \beta \Big( 1 - \frac{1}{e^2} \Big)e^{\frac{2t+c}{\beta}}
    ~.
\end{align*}

Dividing both sides by $\beta$, we need to verify the non-negativeness of:
\[
    \frac{8t + 2c}{\beta} + 2 - \Big( 1 - \frac{1}{e^2} \Big)e^{(2t+c)/\beta}~.
\]

By $\beta \geq t+c$, we have:
\begin{align*}
    \frac{2t+c}{\beta}
    &
    = \frac{4t + c + 2(t+\beta)}{3\beta} 
    \le \frac{1}{3} \Big( \frac{4t+c}{\beta} + 2 \Big)
    ~.
\end{align*}

Changing variables with $x = (4t+c)/s \in (0,4]$, the inequality becomes:
\[
    2x + 2 - \Big(1 - 1/e^2\Big) e^{(x+2)/3} ~\ge~ 0
    ~.
\]

This is concave in $x$, and holds at $x = 0$ (left $\approx 0.316 > 0$) and $x = 4$ (left $\approx 3.611 > 0$).

\end{document}